\documentclass[times,final]{elsarticle}

\usepackage{framed,multirow}

\usepackage{amssymb}
\usepackage{latexsym}
\usepackage[margin=2.5cm]{geometry}

\usepackage{amsmath,mathrsfs}
\usepackage{amsthm}
\usepackage[ruled,vlined]{algorithm2e}
\usepackage{graphicx}
\usepackage{mathtools}
\usepackage{float}
\usepackage[normalem]{ulem}
\usepackage[overload]{empheq}

\DeclarePairedDelimiter\floor{\lfloor}{\rfloor}
\newtheorem{lemma}{Lemma}[section]
\newtheorem{theorem}{Theorem}[section]
\newtheorem{remark}{Remark}[section]

\DeclareMathOperator{\he}{He}
\DeclareMathOperator{\spn}{span}
\DeclareMathOperator{\sgn}{sgn}
\DeclareMathOperator{\erf}{erf}
\usepackage{chngcntr}
\counterwithin{table}{section}
\counterwithin{figure}{section}

\usepackage{url}
\usepackage{xcolor}
\definecolor{newcolor}{rgb}{.8,.349,.1}

\begin{document}


\begin{frontmatter}

  \title{The Wigner Function of Ground State and One-Dimensional Numerics}%
  
  %
  \author[1]{Hongfei Zhan}
  \author[2]{Zhenning Cai\corref{cor1}}
  \cortext[cor1]{Corresponding author.
    E-mail address: matcz@nus.edu.sg}
  \author[1,3,4]{Guanghui Hu}
  %
  \address[1]{Department of Mathematics, Faculty of Science and Technology, University of Macau, Macao SAR, China}
  \address[2]{Department of Mathematics, National University of Singapore, 10 Lower Kent Ridge Road, Singapore}
  \address[3]{Zhuhai UM Science \& Technology Research Institute, Zhuhai, Guangdong, China}
  \address[4]{Guangdong-Hong Kong-Macao Joint Laboratory for Data-Driven Fluid Mechanics and Engineering Applications, University of Macau, Macau 999078, China}
  %

  \begin{abstract}
    In this paper, the ground state Wigner function of a many-body
    system is explored theoretically and numerically. First, an
    eigenvalue problem for Wigner function is derived based on the
    energy operator of the system. The validity of finding the ground
    state through solving this eigenvalue problem is obtained by
    building a correspondence between its solution and the solution of
    stationary Schr\"odinger equation. Then, a numerical method is
    designed for solving proposed eigenvalue problem in one
    dimensional case, which can be briefly described by i) a
    simplified model is derived based on a quantum hydrodynamic model
    [Z. Cai et al, J. Math. Chem., 2013] to reduce the dimension of
    the problem, ii) an imaginary time propagation method is designed
    for solving the model, and numerical techniques such as solution
    reconstruction are proposed for the feasibility of the
    method. Results of several numerical experiments verify our
    method, in which the potential application of the method for large
    scale system is demonstrated by examples with density functional
    theory.
  \end{abstract}


\end{frontmatter}

\section{Introduction}

The study of the ground state of a many-body quantum system plays an
important role in a variety of areas such as geometry optimization of
molecules, photon absorption spectra of atoms and molecules, and linear
response theory in the molecular dynamics.

The ground state can be obtained by solving the fundamental governing
equation, i.e., Schr\"odinger equation, in quantum mechanics. However,
due to the curse of dimensionality, direct
numerical study of Schr\"odinger equation via classical mesh-based approaches is intractable even for small molecule such
as methane. Hence, approximate solution of Schr\"odinger equation
has been a long-standing research topic, in which many pioneer works
have been done. For example, quantum Monte Carlo methods
\cite{acioli1997,yan2017} uses stochastic methods to evaluate
integrals arising in the many-body problems. Quantum Monte Carlo
method offers potential to describe directly many-body effect of a
quantum system. However, the efficiency of the method suffers
from its slow convergence in the simulations. By using a single Slater
determinant as an ansatz for the many-body wavefunction, the main task
of Hartree-Fock method\cite{echenique2007} is to solve a set of
equations derived from a variational method. With acceleration
techniques for the self-consistent field iteration and quality solver
for the generalized eigenvalue problem, the efficiency of Hartree-Fock
method becomes acceptable for large scale system from practical
problems, which makes the method very popular in the quantum
computational chemistry community even nowadays. However, the lack of
electron correlation would introduce large deviations from the
experimental results, which limits the application of the Hartree-Fock
method. Density functional theory\cite{kohn1998,medvedev2017} combines
the advantages from above two
methods. Theoretically\cite{Hohenberg-Kohn}, it has been proved that
three dimensional ground state electron density is a fundamental
quantity in a given many-body system, and both the electron exchange
and the electron correlation are described in derived Kohn-Sham
model\cite{Kohn-Sham}. Numerically, the techniques developed for
Hartree-Fock method can be borrowed for solving Kohn-Sham model. Even
better, with the application of local basis functions for the
wavefunction\cite{BaoHuLiu2012, lin2012}, the numerical efficiency can
potentially be further improved by using fast solvers for sparse
system.

Besides the conventional Schr\"odinger wave function in Hilbert space,
the Wigner phase-space quasi-distribution function \cite{wigner1932}
provides an equivalent approach to describe quantum object that bears
a close analogy to classical mechanics \cite{zachos2002}. Moreover, the intriguing mathematical structure of the Weyl-Wigner correspondence has also been employed in some advanced topics, such as the
deformation quantization\cite{zachos2014}. The Wigner
formalism has been applied to a variety of situations ranging from
atomic physics \cite{vacchini2007} to quantum electronic transport
\cite{schwaha2013,weinbub2018} and many-body quantum
  systems\cite{sellier2015Introduction}. Furthermore, both pure and
mixed states of a quantum system can be handled in a unified approach
by Wigner function\cite{william2008}. All theoretical advantages motivate the
research on developing models and numerical methods for finding the
Wigner functions of a quantum system. 

Different from the situation for Schr\"odinger equation that there
have been lots of mature approximate models and numerical methods,
more efforts are needed towards the Wigner functions.  The first
attempts to simulate quantum phenomena by Wigner function were
\cite{frensley1987,frensley1990} for one-dimensional one-body
case. Recently, several methods were designed for the simulation based
on Wigner function, such as cell average spectral element
method\cite{shao2011}, moment
method\cite{li2014Hyperbolic,furtmaier2016},
WENO-solver\cite{dorda2015}, Gaussian beam method\cite{yin2013}, etc. While there were also various
stochastic methods, e.g., signed particle Wigner Monte Carlo
method\cite{nedjakov2004,nedjalkov2013,sellier2014Benchmark} and path
integral method\cite{larkin2016,larkin2017,amartya2019}. In many-body
situation, the Wigner based simulation was achieved by
advective-spectral-mixed method\cite{shao2016AdvectiveSpectralMixed},
Monte Carlo
method\cite{sellier2014ManyBody,sellier2015Fermion,sellier2016Full}
and the method based on branching random walk\cite{shao2020}.  It is
known that in a dynamic study of a given system, an initial state of
the system should be specified, which is the ground state of the
system in most cases. It is noted that although there have been works
mentioned above for dynamics of a given quantum system, the work
towards the Wigner functions of the ground state is rare. To our best
knowledge, only \cite{sellier2014} proposed a feasible framework to
handle both time-dependent and time-independent problem based on Monte
Carlo method, while no result on the deterministic method for a many
body quantum system can be found from the literature, even for the
simplest one-dimensional two-body case.

In this paper, in the category of deterministic approach, with the aid of density functional theory, the ground
state Wigner function is explored both theoretically and numerically
for a given many-body quantum system. 
More specifically, we firstly
derive an eigenvalue problem of energy operator for Wigner function
based on the stationary Wigner equation. Then the correspondence
between Wigner eigenfunction and Schr\"odinger eigenfunction in the sense of construction is deduced to guarantee the
validity of calculation of the ground state Wigner function through
solving this eigenvalue problem. Focusing on the one-dimensional case,
a numerical method based on imaginary time propagation
method \cite{chin2009,lehtovaara2007,philipp2013} is designed for the solution of the proposed
eigenvalue problem. A quantum hydrodynamic model proposed in \cite{cai2013}
is simplified in our work to reduce the dimension of the problem,
while a reconstruction method is proposed to resolve the
well-posedness issue introduced by truncating the approximation. Four
examples are tested to show the effectiveness of our method. The
numerical convergence of the method can be observed clearly in all
numerical experiments. Furthermore, the capability of the numerical
method on calculating the excited states of the system, and on
handling the case with singular potential, is also demonstrated
successfully in the harmonic oscillator and the hydrogen examples,
respectively. More importantly, the potential of our method for the
ground state calculation of large-scale systems is also shown
obviously in the last two examples with effective potentials in
density functional theory. It is worth mentioning that compared with
the Monte Carlo approach, our method is more robust in the sense that
random initial guess of the Wigner function is adopted in our
simulation. It is known that a general issue of Monte Carlo method is
its slow convergence, while our method successfully demonstrates the
theoretical convergence rate of linear finite element method. The
method can be extended to higher-order cases in a natural way.

The rest of this paper is organized as follows: In Section
\ref{sec-wigner-background} we briefly introduce Wigner function
and stationary Wigner equation. The derivation of the
eigenvalue problem and the correspondence between Wigner eigenfunction and Schrodinger eigenfunction are
demonstrated in Section \ref{sec:eig}. In Section
\ref{sec:hermite expansion}, the discretization along $\mathbf{p}$ direction is provided. One-dimensional numerics based on imaginary time propagation method is considered in Section \ref{sec:itp}. Four numerical examples are presented in
Section \ref{sec:numerical result}. In Section \ref{sec:conclusions}
we conclude this paper and discuss the direction of future work. The
conversion between wave function and the coefficient functions and
related numerical benefit are illustrated in appendix A.

For convenience, we only consider the Hartree atomic units $\hbar=m=e=1$, where
$\hbar $ is the reduced Planck constant, $m$ is the
effective mass of electron, and $e$ is the positive electron
charge. The range of appearing integrals is from $-\infty$ to
$\infty$ without extra explanation.

\section{Wigner function and stationary Wigner equation}
\label{sec-wigner-background}

In this section, the Wigner formalism and the widely used stationary Wigner equation will be introduced briefly.

We define the density matrix
\begin{equation}
  \rho(\mathbf{x},\mathbf{x}')=\sum\limits_jP_j\psi_j(\mathbf{x})\psi_j^*(\mathbf{x}'),
\end{equation}
where $\psi_j$ is the $j$-th eigenfunction of the time-independent
Schr\"odinger equation
\begin{equation}\label{eqn:psi-time-independent-schrodinger}
  H_s\psi_j(\mathbf{x})
  =\left[-\frac{\nabla^2}{2}+V(\mathbf{x})\right]\psi_j(\mathbf{x})
  =E_j\psi_j(\mathbf{x}),
\end{equation}
and $P_j$ is the probability to find the $j$-th eigenstate. The Wigner
function $f(\mathbf{x},\mathbf{p})$ in the $2D$-dimensional phase space $(\mathbf{x},\mathbf{p})\in\mathbb{R}^{2D}$ is defined by
applying Wigner-Weyl transform to the density matrix
\begin{equation}\label{eqn:wigner-definition}
  f(\mathbf{x},\mathbf{p})=\frac{1}{(2\pi)^D}\int\rho\left(\mathbf{x}+\frac{\mathbf{y}}{2},\mathbf{x}-\frac{\mathbf{y}}{2}\right)\exp(-i\mathbf{p}\cdot\mathbf{y})d\mathbf{y}.
\end{equation}
By taking integral of
(\ref{eqn:wigner-definition}) with respect to $\mathbf{p}$, we get the particle density
\begin{equation}
  \rho(\mathbf{x})= \rho(\mathbf{x}, \mathbf{x}) = \int f(\mathbf{x},\mathbf{p})d\mathbf{p}.
\end{equation}
Following the basic property of the Weyl transform, the energy in the Wigner formalism can be expressed by
\begin{equation}
  E=\iint \left[\frac{|\mathbf{p}|^2}{2}+V(\mathbf{x}) \right]f(\mathbf{x},\mathbf{p})d\mathbf{x}d\mathbf{p}.
\end{equation}
Specially, if the density matrix corresponds to a pure state $e^{iA}\psi(\mathbf{x})$, where $\psi(\mathbf{x})$ is a real-valued function, with the fact that Wigner function is also a real-valued function, we can derive that
\begin{equation}
  \begin{aligned}
    f(\mathbf{x},-\mathbf{p})
    &=f(\mathbf{x},-\mathbf{p})^*\\
    &=\frac{1}{(2\pi)^D}\int\psi\left(\mathbf{x}+\frac{-\mathbf{y}}{2}\right)\psi^*\left(\mathbf{x}-\frac{-\mathbf{y}}{2}\right)\exp(-i\mathbf{p}\cdot(-\mathbf{y}))d\mathbf{y}
    =f(\mathbf{x},\mathbf{p}).
  \end{aligned}
\end{equation}
Thus the Wigner function of any pure state with a constant phase factor is even with respect to $\mathbf{p}$.

Following the method in \cite{hillery1984}, we can deduce  the stationary Wigner equation from the time-independent Schr\"odinger equation
\begin{equation}\label{eqn:stationary-wigner-eqn}
  \mathbf{p}\cdot\nabla_\mathbf{x}f(\mathbf{x},\mathbf{p})
  +(\Theta[V]f )(\mathbf{x},\mathbf{p})
  =0,
\end{equation}
where $\Theta[V]f $ is a non-local pseudo-differential operator defined by
\begin{equation}\label{eqn:pseudo-differential-operator}
  (\Theta[V]f )(\mathbf{x},\mathbf{p})
  =\int V_{w}(\mathbf{x},\mathbf{p}')f(\mathbf{x},\mathbf{p}-\mathbf{p}')d\mathbf{p}',
\end{equation}
and the Wigner potential reads
\begin{equation}\label{eqn:pseudo-differential-operator-Vw}
  V_w(\mathbf{x},\mathbf{p})
  =\frac{i}{(2\pi)^D}\int \left[V\left(\mathbf{x}+\frac{\mathbf{y}}{2}\right)-V\left(\mathbf{x}-\frac{\mathbf{y}}{2}\right)\right]\exp(-i\mathbf{p}\cdot \mathbf{y})d\mathbf{y}.
\end{equation}
If the potential $V\in C^\omega({\mathbb{R}^D})$
, (\ref{eqn:pseudo-differential-operator}) can be locally expressed by means of Taylor expansion
\begin{equation}\label{eqn:pseudo-differential-taylor}
  (\Theta[V]f )(\mathbf{x},\mathbf{p})
  =-\sum\limits_{\lambda,|\lambda|\text{odd}}\frac{1}{\lambda!(2i)^{|\lambda|-1}}\frac{\partial^\lambda V}{\partial \mathbf{x}^\lambda}\frac{\partial^\lambda}{\partial \mathbf{p}^\lambda}f(\mathbf{x},\mathbf{p}),
\end{equation}
where $\mathbf{\lambda}$ is a $D$-dimensional multi-index, $\mathbf{\lambda}!=\prod_{j=1}^D\lambda_j!$, $\mathbf{x}^\mathbf{\lambda}=\prod_{j=1}^Dx_j^{\lambda_j}$, and
\begin{equation}
  \frac{\partial^\mathbf{\lambda}}{\partial \mathbf{x}^\mathbf{\lambda}}
  =\prod\limits_{j=1}^D\frac{\partial^{\lambda_j}}{\partial x_j^{\lambda_j}},
  ~~~
  \frac{\partial^\mathbf{\lambda}}{\partial \mathbf{p}^\mathbf{\lambda}}
  =\prod\limits_{j=1}^D\frac{\partial^{\lambda_j}}{\partial p_j^{\lambda_j}}.
\end{equation}
We close this section by the following theorem, which illustrates the effect of stationary Wigner equation from the perspective of Fourier transform:
\begin{theorem}\label{thm:stationary wigner eqn-pure state}
	If the eigenenergy of Schr\"odinger Hamiltonian is nondegenerate, the Wigner function $f(\mathbf{x},\mathbf{p})$ satisfying (\ref{eqn:stationary-wigner-eqn}) can be expressed by linear combination of the ones corresponding to pure state.
\end{theorem}
\begin{proof}
	We rewrite Wigner function as follows
	\begin{equation}\label{eqn:definition-tildef}
		f(\mathbf{x},\mathbf{p})=\frac{1}{(2\pi)^D}\int\tilde{f}(\mathbf{x},\mathbf{y})\text{e}^{-i\mathbf{p}\cdot\mathbf{y}}d\mathbf{y}.
	\end{equation}
	Introducing change of variables $\mathbf{u}=\mathbf{x}+\mathbf{y}/2$, $\mathbf{v}=\mathbf{x}-\mathbf{y}/2$, it follows (\ref{eqn:stationary-wigner-eqn}) that
	\begin{equation}
		\left[H_s(\mathbf{u})-H_s(\mathbf{v})\right]\tilde{f}=0,
		~~~\text{where}~~~
		H_s(\mathbf{u})=-\frac{1}{2}\nabla_\mathbf{u}^2+V(\mathbf{u}).
	\end{equation}
	Consequently, in nondegenerate case we have
	Following the equation above we find
	\begin{equation*}
	\begin{aligned}
	  \left(\tilde{f},\psi_i^*(\mathbf{u})\psi_j(\mathbf{v})\right)
	    &=\frac{1}{E_i}\left(\tilde{f},H_s(\mathbf{u})\psi_i^*(\mathbf{u})\psi_j(\mathbf{v})\right)
	    =\frac{1}{E_i}\left(H_s(\mathbf{u})\tilde{f},\psi_i^*(\mathbf{u})\psi_j(\mathbf{v})\right)\\
	    &=\frac{1}{E_i}\left(H_s(\mathbf{v})\tilde{f},\psi_i^*(\mathbf{u})\psi_j(\mathbf{v})\right)
	    =\frac{1}{E_i}\left(\tilde{f},\psi_i^*(\mathbf{u})H_s(\mathbf{v})\psi_j(\mathbf{v})\right)
	    =\frac{E_j}{E_i}\left(\tilde{f},\psi_i^*(\mathbf{u})\psi_j(\mathbf{v})\right)
	    =\delta_{ij}\left(\tilde{f},\psi_i^*(\mathbf{u})\psi_i(\mathbf{v})\right),
	\end{aligned}
	\end{equation*}
    where $\delta_{ij}$ is the Kronecker delta symbol.
\end{proof}

\section{Eigenvalue problem}
\label{sec:eig}
In order to find the Wigner function of ground state, the Wigner analogy of the eigenvalue problem with respect to energy operator is needed. Starting from the eigenvalue problem proposed in \cite{hillery1984}, Lemma \ref{lem:1} is deduced for deriving Wigner eigenvalue problem. Finally, to guarantee the validity of the ground state calculation based on Wigner function, the correspondence between Schr\"odinger eigenfunction and Wigner eigenfunction is established in Theorem \ref{thm:1}.

\begin{lemma}\label{lem:1}
  Let $A(\mathbf{x},\mathbf{p})=F(\mathbf{x})$ be the Weyl transform of the operator $\hat{A}=A(\hat{\mathbf{x}},\hat{\mathbf{p}})=\hat{F}=F(\hat{\mathbf{x}})$, the eigenvalue problem of Wigner function with respect to $A$ is
  \begin{equation}\label{eqn:lemma1-result1}
    \frac{1}{\pi^D}\iint F(\mathbf{v})f(\mathbf{x},\mathbf{r})\text{e}^{2i(\mathbf{v}-\mathbf{x})\cdot(\mathbf{r}-\mathbf{p})}d\mathbf{v}d\mathbf{r}
    =\lambda f(\mathbf{x},\mathbf{p}).
  \end{equation}
  If $F$ is a polynomial, then (\ref{eqn:lemma1-result1}) can be further simplified as
  \begin{equation}\label{eqn:lemma1-result2}
    F\left(\frac{i}{2}\frac{\partial}{\partial\mathbf{p}}+\mathbf{x}\right)f(\mathbf{x},\mathbf{p})
    =\lambda f(\mathbf{x},\mathbf{p}).
  \end{equation}
  For $A(\mathbf{x},\mathbf{p})=G(\mathbf{p})$, similarly, we have the corresponding eigenvalue problem
  \begin{equation}\label{eqn:lemma1-result3}
    \frac{1}{\pi^D}\iint G(\mathbf{r})f(\mathbf{v},\mathbf{p})\text{e}^{-2i(\mathbf{v}-\mathbf{x})\cdot(\mathbf{r}-\mathbf{p})}d\mathbf{v}d\mathbf{r}
    =\lambda f(\mathbf{x},\mathbf{p}),
  \end{equation}
  and if $G$ is polynomial, we have
  \begin{equation}\label{eqn:lemma1-result4}
    G\left(-\frac{i}{2}\frac{\partial}{\partial\mathbf{x}}+\mathbf{p}\right)f(\mathbf{x},\mathbf{p})
    =\lambda f(\mathbf{x},\mathbf{p}).
  \end{equation}
\end{lemma}
In \cite{hillery1984}, the general eigenvalue problem for an operator $\hat{A} = A(\hat{\mathbf{x}},\hat{\mathbf{p}})$ has been formulated as
\begin{equation}
	\left(\frac{4}{\pi^2}\right)^D\iiiint A(\mathbf{y}+\mathbf{y}',\mathbf{q}-\mathbf{q}')f(\mathbf{y}-\mathbf{y}',\mathbf{q}+\mathbf{q}')\exp\{4i\mathbf{y}'\cdot(\mathbf{q}-\mathbf{p})+4i\mathbf{q}'\cdot(\mathbf{y}-\mathbf{x})\}d\mathbf{y}d\mathbf{y}'d\mathbf{q}d\mathbf{q}'
	=\lambda f(\mathbf{x},\mathbf{p}),
\end{equation}
from which \eqref{eqn:lemma1-result1} and \eqref{eqn:lemma1-result3} can be immediately obtained. The alternative forms \eqref{eqn:lemma1-result2} and \eqref{eqn:lemma1-result4} can then be derived via integration by parts. Now we consider the Weyl transform of the energy operator $H(\mathbf{x},\mathbf{p})=|\mathbf{p}|^2/2+V(\mathbf{x})$. By Lemma \ref{lem:1}, using (\ref{eqn:stationary-wigner-eqn}) to cancel the imaginary part we can define the Wigner analogy of the eigenvalue problem corresponding to energy operator
\begin{equation}\label{eqn:wigner-eigenvalue-problem}
  \begin{aligned}
    H_{\rm w}f(\mathbf{x},\mathbf{p})
    &=\frac{1}{2}\left(-\frac{1}{4}\nabla_\mathbf{x}^2f(\mathbf{x},\mathbf{p})+|\mathbf{p}|^2f(\mathbf{x},\mathbf{p})
    +\int V_{\rm eig}(\mathbf{x},\mathbf{p}')f(\mathbf{x},\mathbf{p}-\mathbf{p}')d\mathbf{p}'\right)
    =Ef(\mathbf{x},\mathbf{p}),
  \end{aligned}
\end{equation}
where
\begin{equation}\label{eqn:eig-Veig-def}
  V_{\rm eig}(\mathbf{x},\mathbf{p})
  =\frac{1}{(2\pi)^D}\int\left[V\left(\mathbf{x}+\frac{\mathbf{y}}{2}\right)+V\left(\mathbf{x}-\frac{\mathbf{y}}{2}\right)\right]\exp(-i\mathbf{p}\cdot \mathbf{y})d\mathbf{y}.
\end{equation}
If the potential $V$ is analytic on $\mathbb{R}^D$, the integral in (\ref{eqn:wigner-eigenvalue-problem}) has local expression similar to (\ref{eqn:pseudo-differential-taylor})
\begin{equation}\label{eqn:H-wigner-original-part2-taylor}
  \begin{aligned}
    \int V_{\rm eig}(\mathbf{x},\mathbf{p}')f(\mathbf{x},\mathbf{p}-\mathbf{p}')d\mathbf{p}'
    =\sum\limits_{\mathbf{\lambda},|\mathbf{\lambda}|\text{even}}
    \frac{2}{\mathbf{\lambda}!(2i)^{|\mathbf{\lambda}|}}\frac{\partial^\mathbf{\lambda} V}{\partial \mathbf{x}^\mathbf{\lambda}}
    \frac{\partial^\mathbf{\lambda}}{\partial \mathbf{p}^\mathbf{\lambda}}f(\mathbf{x},\mathbf{p}).
  \end{aligned}
\end{equation}
The eigenvalue problem (\ref{eqn:wigner-eigenvalue-problem}) and stationary Wigner equation (\ref{eqn:stationary-wigner-eqn}) are exactly the real and imaginary parts of the ``star-genvalue problem'' introduced in \cite{thomas1998}, respectively. It is proved in \cite{thomas1998} that the Wigner function satisfying star-genvalue problem corresponds to a pure state wave function.

\begin{remark}\label{rem:eigenfunction-Hw}
  Suppose the validity of separation of variables for the eigenfunctions of $H_w$. Following a similar discussion to Theorem \ref{thm:stationary wigner eqn-pure state}, we can derive that
  \begin{equation}
    \left[H_s(\mathbf{u})+H_s(\mathbf{v})\right]\tilde{f}=2E\tilde{f}.
  \end{equation}
  Using the method of separation of variables we find
  \begin{equation*}
      \frac{H_s(\mathbf{u})U}{U}+\frac{H_s(\mathbf{v})V}{V}=2E,
  \end{equation*}
  Therefore the eigenfunction has the form
  \begin{equation}
    f_{ij}(\mathbf{x},\mathbf{p})=\frac{1}{(2\pi)^D}\int\psi_i^*\left(\mathbf{x}+\frac{\mathbf{y}}{2}\right)\psi_j\left(\mathbf{x}-\frac{\mathbf{y}}{2}\right)\textrm{e}^{-i\mathbf{p}\cdot\mathbf{y}}d\mathbf{y}.
  \end{equation}
  with eigenvalue $E_{ij}=(E_i+E_j)/2$.
\end{remark}
\begin{remark}
  To distinguish the Wigner function corresponding to pure state from other eigenfunctions in Remark \ref{rem:eigenfunction-Hw}, we have following relation:
  \begin{equation}
    \iint f_{ij}(\mathbf{x},\mathbf{p})d\mathbf{x}d\mathbf{p}
    =\int\psi_i^*(\mathbf{x})\psi_j(\mathbf{y})d\mathbf{x}
    =\delta_{ij},
  \end{equation}
  where $\delta_{ij}$ is Dirac delta function.
\end{remark}
Furthermore, the correspondence between Schr\"odinger eigenfunction and Wigner eigenfunction can be derived as follows.
\begin{theorem}\label{thm:1}
  Let $\psi(\mathbf{x})$ be a Schr\"odinger eigenfunction with energy $E$, then there exists a Wigner eigenfunction $f(\mathbf{x},\mathbf{p})$ with the same eigenvalue $E$ satisfying the stationary Wigner equation. Conversely, if $f(\mathbf{x},\mathbf{p})$ is a Wigner function satisfying stationary Wigner equation and the eigenvalue problem, there exists a Schr\"odinger eigenfunction function with the same eigenvalue.
\end{theorem}
\begin{proof}
  The derivation of the eigenvalue problem (\ref{eqn:wigner-eigenvalue-problem}) and stationary Wigner equation (\ref{eqn:stationary-wigner-eqn}) has clearly shown that the Wigner-Weyl transform of $\psi$ satisfies these two equations.
  
  On the contrary, let $f(\mathbf{x},\mathbf{p})$ be an eigenfunction of $H_{\rm w}$ with eigenvalue $E$, which satisfies (\ref{eqn:stationary-wigner-eqn}). Applying the inverse Wigner-Weyl transform we have
  \begin{equation}\label{eqn:psi-inverse-wigner-transform}
    \begin{aligned}
      \psi(\mathbf{x})
      =A\int f\left(\frac{\mathbf{x}+\mathbf{x}_0}{2},\mathbf{p}\right)\text{e}^{i\mathbf{p}\cdot(\mathbf{x}-\mathbf{x}_0)}d\mathbf{p},
    \end{aligned}
  \end{equation}
  where $A$ is the normalization constant. It can be verified by direct calculation and substitution of (\ref{eqn:stationary-wigner-eqn}) that
  \begin{equation}
    \left(-\frac{\nabla^2}{2}+V(\mathbf{x})\right)\psi
    =E\psi,
  \end{equation}
  which means such construction yields a Schr\"odinger eigenfunction with the same energy.
\end{proof}
The above theorem guarantees the 
correspondence between Schr\"odinger eigenfunction and Wigner eigenfunction in the sense of construction, which validates the ground state calculation based on Wigner function by solving the eigenvalue problem with the constraint of stationary Wigner equation. Furthermore, based on the construction in theorem \ref{thm:1}, the conversion between Schr\"odinger eigenfunction and Wigner eigenfunction expressed by coefficient functions is established in \ref{app:conversion}. Below we will focus on the Wigner formalism to develop our numerical schemes.

\section{Hermite expansion of Wigner function}
\label{sec:hermite expansion}
To deal with the global integral operators and reduce the dimension of problem, we consider the quantum hydrodynamic model proposed in \cite{cai2013}. It is noted that the hydrodynamic model in \cite{cai2013} is developed for time-dependent simulations, and is derived based on a series expansion of the Wigner function with respect to the momentum variable $\mathbf{p}$, where the basis functions vary for different positions and times. Here in the computation of ground states, we fix the bases and expand the Wigner function as
%
\begin{equation}\label{eqn:wigner-hermite-expansion}
  f(\mathbf{x},\mathbf{p})
  =\sum\limits_{\mathbf{\alpha}\in\mathbb{N}^D}f_\mathbf{\alpha}(\mathbf{x})\mathcal{H}_\mathbf{\alpha}(\mathbf{p}),
\end{equation}
where $\mathbf{\alpha}$ is a $D$-dimensional multi-index. The basis function $\mathcal{H}_\mathbf{\alpha}$ is defined as
\begin{equation}\label{eqn:hermite-expansion-basis-function}
  \mathcal{H}_\mathbf{\alpha}(\mathbf{p})\
  =\frac{1}{(2\pi)^{D/2}}\exp\left(-\frac{|\mathbf{p}|^2}{2}\right)\prod\limits_{j=1}^D\he_{\alpha_j}(p_j),
\end{equation}
where $\he_n(x)$ is the $n$-degree Hermite polynomial
\begin{equation}\label{eqn:hermite-polynomial-definition}
  \he_n(x)
  =(-1)^n\exp\left(\frac{x^2}{2}\right)\frac{d^n}{dx^n}\exp\left(-\frac{x^2}{2}\right)
  =n!\sum\limits_{m=0}^{\floor{\frac{n}{2}}}\frac{(-1)^m}{m!(n-2m)!}\frac{x^{n-2m}}{2^m}.
\end{equation}
To deduce the equations of coefficient functions corresponding to stationary Wigner equation and the eigenvalue problem, we need the following useful properties of Hermite polynomial \cite{handbook}:
\begin{enumerate}[1.]
\item Orthogonality: $\int \he_m(x)\he_n(x)\exp(-x^2/2)dx=m!\sqrt{2\pi}\delta_{m,n}$;
\item Recursion relation: $\he_{n+1}(x)=x\he_n(x)-n\he_{n-1}(x)$;
\item Differential relation: $\he_n'(x)=n\he_{n-1}(x)$.
\end{enumerate}
Combining the last two relations we find
\begin{equation}\label{eqn:relation-hermite}
  \left[\he_n(x)\exp(-x^2/2)\right]'
  =-\he_{n+1}(x)\exp(-x^2/2).
\end{equation}
Therefore
\begin{equation}\label{eqn:derivative-H_alpha}
  \frac{\partial}{\partial p_j}\mathcal{H}_\mathbf{\alpha}(\mathbf{p})
  =-\mathcal{H}_{\mathbf{\alpha}+\mathbf{e}_j}(\mathbf{p}).
\end{equation}
A direct result of the orthogonality of Hermite polynomials is
\begin{equation}\label{eqn:coefficient-function-calculation}
  f_\mathbf{\alpha}(\mathbf{x})
  =\frac{1}{\mathbf{\alpha}!}\int \left(\prod\limits_{j=1}^D\he_{\alpha_j}(p_j)\right)f(\mathbf{x},\mathbf{p})d\mathbf{p}.
\end{equation}
Due to the fact that the Wigner function of any pure state is even with respect to $\mathbf{p}$, in the situation of ground state calculation, it is reasonable to assume
\begin{equation}
  f_\mathbf{\alpha}(\mathbf{x})=0,
  ~~~\text{if }|\mathbf{\alpha}|\text{ is odd}.
\end{equation}
Finally, to avoid a system with infinite unknowns, we only consider $f_\mathbf{\alpha}(\mathbf{x})$ with $|\mathbf{\alpha}|\leqslant M$, where $M$ is an even positive number.

Particularly, for two Wigner functions $f_a(\mathbf{x},\mathbf{p})$ and $f_b(\mathbf{x},\mathbf{p})$ corresponding to two orthogonal eigenstates $\psi_a$ and $\psi_b$, respectively, following the property of Weyl transform, using integration by part we obtain (For detailed illustration, one can refer to \cite{william2008}) 
\begin{equation}\label{eqn:orthogonality}
  \begin{aligned}
    &(2\pi)^{-D}\sum\limits_{\mathbf{\alpha},\mathbf{\beta}\leqslant M}C_{\mathbf{\alpha}+\mathbf{\beta}}\int f_\mathbf{\alpha}^{(a)}(\mathbf{x})f_\mathbf{\beta}^{(b)}(\mathbf{x})d\mathbf{x}
    \approx\iint f_a(\mathbf{x},\mathbf{p})f_b(\mathbf{x},\mathbf{p})d\mathbf{x}d\mathbf{p}
    =(2\pi)^{-D}|\langle\psi_a|\psi_b\rangle|^2
    =0,
  \end{aligned}
\end{equation}
where $f_\mathbf{\alpha}^{(a)}(\mathbf{x})$ and $f_\mathbf{\beta}^{(b)}(\mathbf{x})$ are the coefficient functions of $f_a(\mathbf{x},\mathbf{p})$ and $f_b(\mathbf{x},\mathbf{p})$, respectively, and (For detailed derivation, see \ref{app:coefficient}.)
\begin{equation}\label{eqn:C_aph-int-H_aphH_bet}
  C_\mathbf{\alpha}
  =\int\left(\prod_{j=1}^D\he_{\alpha_j}(p_j)\right)\exp(-|\mathbf{p}|^2)d\mathbf{p}
  =\left\{\begin{array}{ll}
  \displaystyle\left(-1/4\right)^{|\mathbf{\alpha}|/2}\frac{\mathbf{\alpha}!}{(\mathbf{\alpha}/2)!}\pi^{D/2},
  &\text{if $\alpha_j$ is even for $1\leqslant j\leqslant D$},\\
  0,	&\text{otherwise}.
  \end{array}\right.
\end{equation}
Then we will use the properties mentioned above to derive governing equations of the coefficient functions corresponding to the stationary Wigner equation and the eigenvalue problem. The procedure follows the Petrov-Galerkin spectral method, which can also be regarded as taking moments on both sides of the equations.

\subsection{Stationary Wigner equation}
For convenience, we take $f_\mathbf{\alpha}=0$ for any multi-index $\mathbf{\alpha}$ with at least one negative component. Following the same method in \cite{cai2012}, we get the series expansion fo the first term in (\ref{eqn:stationary-wigner-eqn}) as
\begin{equation}\label{eqn:stationary-wigner-eqn-term1st-hermite}
  \mathbf{p}\cdot\nabla_\mathbf{x}f(\mathbf{x},\mathbf{p})
  =\sum\limits_\mathbf{\alpha}\sum\limits_{j=1}^D\left((\alpha_j+1)\frac{\partial f_{\mathbf{\alpha}+\mathbf{e}_j}}{\partial x_j}+\frac{\partial f_{\mathbf{\alpha}-\mathbf{e}_j}}{\partial x_j}\right)\mathcal{H}_\mathbf{\alpha}(\mathbf{p}).
\end{equation}
Using (\ref{eqn:relation-hermite}), the pseudo-differential operator term $\Theta[V]f $ with expression (\ref{eqn:pseudo-differential-taylor}) becomes
\begin{equation}\label{eqn:stationary-wigner-eqn-term2nd-hermite}
  (\Theta[V]f )(\mathbf{x},\mathbf{p})
  =\sum\limits_\mathbf{\alpha}\bigg{(}\sum\limits_{\mathbf{\lambda}\leqslant\mathbf{\alpha},|\mathbf{\lambda}|\text{odd}}\frac{1}{\mathbf{\lambda}!(2i)^{|\mathbf{\lambda}|-1}}\frac{\partial^\mathbf{\lambda} V}{\partial \mathbf{x}^\mathbf{\lambda}}
  f_{\mathbf{\alpha}-\mathbf{\lambda}}\bigg{)}\mathcal{H}_\mathbf{\alpha}(\mathbf{p}).
\end{equation}
Plugging (\ref{eqn:stationary-wigner-eqn-term1st-hermite}) and (\ref{eqn:stationary-wigner-eqn-term2nd-hermite}) into (\ref{eqn:stationary-wigner-eqn}) and equating the coefficient of each basis function to zero, we obtain
\begin{equation}\label{eqn:stationary-wigner-eqn-hermite}
  \sum\limits_{j=1}^D\left((\alpha_j+1)\frac{\partial f_{\mathbf{\alpha}+\mathbf{e}_j}}{\partial x_j}+\frac{\partial f_{\mathbf{\alpha}-\mathbf{e}_j}}{\partial x_j}\right)
  +\sum\limits_{\mathbf{\lambda}\leqslant\mathbf{\alpha},|\mathbf{\lambda}|\text{odd}}\frac{1}{\mathbf{\lambda}!(2i)^{|\mathbf{\lambda}|-1}}\frac{\partial^\mathbf{\lambda} V}{\partial \mathbf{x}^\mathbf{\lambda}}f_{\mathbf{\alpha}-\mathbf{\lambda}}
  =0.
\end{equation}
This equation holds for every $\alpha$ with $|\alpha|$ being odd.

\subsection{Eigenvalue problem}
Using the recursion relation of Hermite polynomials, the first two terms in (\ref{eqn:wigner-eigenvalue-problem}) can be expanded as
\begin{equation}
  \begin{aligned}
    \left(-\frac{1}{4}\nabla_\mathbf{x}^2+|\mathbf{p}|^2\right)f(\mathbf{x},\mathbf{p})
    =\sum\limits_\mathbf{\alpha}\bigg{(}-\frac{1}{4}\nabla_\mathbf{x}^2f_\mathbf{\alpha}
    +\sum\limits_{j=1}^D\left((\alpha_j+2)(\alpha_j+1)f_{\mathbf{\alpha}+2\mathbf{e}_j}+(2\alpha_j+1)f_\mathbf{\alpha}+f_{\mathbf{\alpha}-2\mathbf{e}_j}\right)\bigg{)}
    \mathcal{H}_\mathbf{\alpha}(\mathbf{p}).
  \end{aligned}
\end{equation}
With the help of the differential relation, the last term with local expression (\ref{eqn:H-wigner-original-part2-taylor}) can be written as
\begin{equation}
  \int V_{\rm eig}(\mathbf{x},\mathbf{p}')f(\mathbf{x},\mathbf{p}-\mathbf{p}')d\mathbf{p}'
  =\sum\limits_\mathbf{\alpha}\bigg{(}\sum\limits_{\mathbf{\lambda}\leqslant\mathbf{\alpha},|\mathbf{\lambda}|\text{even}}\frac{2}{\mathbf{\lambda}!(2i)^{|\mathbf{\lambda}|}}\frac{\partial^\mathbf{\lambda} V}{\partial \mathbf{x}^\mathbf{\lambda}}f_{\mathbf{\alpha}-\mathbf{\lambda}}\bigg{)}\mathcal{H}_\mathbf{\alpha}(\mathbf{p}).
\end{equation}
Similar to the derivation of (\ref{eqn:stationary-wigner-eqn-hermite}), the orthogonality of Hermite polynomials yields
\begin{equation}\label{eqn:wigner-eigenvalue-problem-hermite}
  \begin{aligned}
    -\frac{1}{4}\nabla_\mathbf{x}^2f_\mathbf{\alpha}
    +\sum\limits_{j=1}^D\left((\alpha_j+2)(\alpha_j+1)f_{\mathbf{\alpha}+2\mathbf{e}_j}+(2\alpha_j+1)f_\mathbf{\alpha}+f_{\mathbf{\alpha}-2\mathbf{e}_j}\right)
    +\sum\limits_{\mathbf{\lambda}\leqslant\mathbf{\alpha},|\mathbf{\lambda}|\text{even}}\frac{2}{\mathbf{\lambda}!(2i)^{|\mathbf{\lambda}|}}\frac{\partial^\mathbf{\lambda} V}{\partial \mathbf{x}^\mathbf{\lambda}}f_{\mathbf{\alpha}-\mathbf{\lambda}}
    =2Ef_\mathbf{\alpha}.
  \end{aligned}
\end{equation}
In the above equation, we choose $\alpha$ such that $|\alpha|$ is even.

So far we have already derived the general eigenvalue problem and the constraint based on stationary Wigner equation for our model. In next section, we will restrict ourselves to one-dimensional case, and provide a feasible numerical framework on the strength of imaginary time propagation method.

\section{One-dimensional numerics}\label{sec:itp}
Now we restrict ourselves to one-dimensional case, and take $M=2K$, 
where $K\in\mathbb{N}$. Taking $\alpha=2k+1$ in (\ref{eqn:stationary-wigner-eqn-hermite}) we find
\begin{equation}\label{eqn:stationary-wigner-eqn-hermite-1D}
(2k+2)\frac{\partial f_{2k+2}}{\partial x}+\frac{\partial f_{2k}}{\partial x}
+\sum\limits_{l =0}^k\frac{1}{(2l+1)!(2i)^{2l}}\frac{\partial^{2l+1}V}{\partial x^{2l+1}}f_{2k-2l}
=0,
~~~k=0,1,\dots,K.
\end{equation}
\begin{remark}
  For $0\leqslant k\leqslant K$, since
  \begin{equation}
	    \int f_{2k}(x)dx=\frac{1}{(2k)!}\iint\he_{2k}(p)f(x,p)dxdp
	    =\frac{1}{(2k)!}\langle\he_{2k}(p)\rangle
  \end{equation}
  must be a finite expectation value, we have
  $\lim_{x\rightarrow-\infty}f_{2k}(x)=0$ for all $k\geqslant0$, 
  integrating both sides of (\ref{eqn:stationary-wigner-eqn-hermite-1D}) from $-\infty$ to $x$ we obtain
  \begin{equation}
	(2k+2)f_{2k+2}(x)+f_{2k}(x)+\sum\limits_{l =0}^k\frac{1}{(2l+1)!(2i)^{2l}}\int_{-\infty}^x\left(\frac{\partial^{2l+1}V}{\partial x^{2l+1}}f_{2k-2l}\right)(s)ds
	=0.
  \end{equation}
  Thus $f_{2k+2}$ can be calculated by $\{f_{2l}\}_{l=0}^k$. Once given the density $\rho(x)=f_0(x)$, we could construct $\{f_{2k}\}_{k=1}^\infty$ recursively. Therefore in one-dimensional case, the Wigner function of pure state is determined by the density.
\end{remark}
The corresponding eigenvalue problem is given by taking $\alpha=2k$ in (\ref{eqn:wigner-eigenvalue-problem-hermite})
\begin{equation}\label{eqn:H-wigner-taylor-simplied-hermite-1D}
\begin{aligned}
-\frac{1}{4}\nabla_x^2f_{2k}
+(2k+2)(2k+1)f_{2k+2}+(4k+1)f_{2k}+f_{2k-2}
+\sum\limits_{l =0}^k\frac{2}{(2l )!(2i)^{2l}}\frac{\partial^{2l}V}{\partial x^{2l}}f_{2k-2l}
=2Ef_{2k},
~~~k=0,1,\dots,K.
\end{aligned}
\end{equation}
It is noted that (\ref{eqn:H-wigner-taylor-simplied-hermite-1D}) depicts an underdetermined eigenvalue system of coupled functions, which is difficult to directly solve by commonly-used method such like block Schur factorization due to its coupled structure. Instead we consider solving it from the perspective of time propagation. In next subsection we will introduce imaginary time propagation method for the utilization of ground state calculation. It is worth to mention that the truncated system is underdetermined due to the appearance of $f_{2K+2}$, a reconstruction approach is proposed in Subsection \ref{subsec:reconstruction} to achieve a well-posed system. Finally, the numerical detail is shown in the last subsection.
\subsection{Imaginary time propagation method}
Imaginary time propagation (ITP) method is a widely-used method for solving eigenvalue problems such as the time-independent Schr\"odinger equation. Its basic idea is to use the fact that for any Hermitian operator $H$ and large $t$, we have $\exp(-tH) \Psi$ is approximately an eigenfunction of $H$ associated with its smallest eigenvalue, given that $\Psi$ is not orthogonal to the corresponding eigenspace. In order to prevent the eigenfunction from being too large/small, renormalization is applied during the time propagation. In the context of finding the ground state, we can illustrate this idea by considering the time-dependent Schr\"odinger equation
\begin{equation}\label{eqn:psi-time-dependent-schrodinger}
  i\frac{\partial\Psi(x,t)}{\partial t}
  =H_{\rm s}\Psi(x,t),
  ~~~\Psi(x,0)=\Psi_0(x),
\end{equation}
where $H_{\rm s}=-\nabla^2/2+V$ is the Hamiltonian in the Schr\"odinger picture. 
We adopt Wick rotation of the time coordinate, i.e., $t=-i\tau$, then (\ref{eqn:psi-time-dependent-schrodinger}) becomes
\begin{equation}\label{eqn:psi-itp}
  -\frac{\partial\Psi(x,\tau)}{\partial \tau}
  =H_s\Psi(x,\tau),
  ~~~
  \Psi(x,0)=\Psi_0(x),
\end{equation}
with the formal solution
\begin{equation}
  \Psi(x,\tau)
  =\text{e}^{-\tau H}\Psi_0(x)
  \rightarrow
  \text{e}^{-E_0\tau}\psi_0(x)
  ~~~\text{as}~~~t\rightarrow\infty.
\end{equation}
Since the ground state eigenfunction $\psi_0$ has the lowest eigenvalue $E_0$, given the initial state that is not orthogonal to $\psi_0$, the normalized steady-state of (\ref{eqn:psi-itp}) yields the ground state of the time-independent Schr\"odinger equation (\ref{eqn:psi-time-independent-schrodinger}) as other exponentials decay more rapidly. To find the excited state, we only need to propagate several functions simultaneously with orthogonalization process after each step.

Following the similar way to derive (\ref{eqn:stationary-wigner-eqn-hermite}), with the help of (\ref{eqn:psi-itp}), direct calculating the time derivative of Wigner function in the expression of (\ref{eqn:wigner-definition}) we obtain the equations for ITP. Recall the eigenvalue problem for the Wigner function (\ref{eqn:H-wigner-taylor-simplied-hermite-1D}), the ITP governing equations for coefficient functions are the ones by replacing the right-hand side of (\ref{eqn:H-wigner-taylor-simplied-hermite-1D}) with negative time derivative:
\begin{equation}\label{eqn:coefficient-funciton-itp}
  \begin{aligned}
    -\frac{1}{4}\nabla_x^2f_{2k}
    +(2k+2)(2k+1)f_{2k+2}+(4k+1)f_{2k}+f_{2k-2}
    +\sum\limits_{l =0}^k\frac{2}{(2l )!(2i)^{2l}}\frac{\partial^{2l}V}{\partial x^{2l}}f_{2k-2l}
    =-\frac{\partial f_{2k}}{\partial \tau},
  \end{aligned}
\end{equation}
where $k=0,1,\dots,K$.

It is desired to mention that the constraint of stationary Wigner equation makes the gap between the smallest eigenvalue and second-smallest eigenvalue larger, i.e., from $(E_0+E_1)/2-E_0$ to $E_1-E_0$, which accelerates the convergence process. Although the infinite system satisfies the constraint of stationary Wigner equation if the initial condition satisfies such constraint. Due to the closure and numerical error arising during time evolution, we enforce this condition at each time step. The evolution of $\{f_{2k}\}_{k=0}^K$ is depicted by the following two-step procedure
\begin{subequations}
  \begin{empheq}[left={\empheqlbrace\,}]{align}
    &[f_0,f_2,\dots,f_{2K}]^T
      =\textrm{e}^{-\tau H_{\rm cf}}[f_0,f_2,\dots,f_{2K},f_{2K+2}]^T;\label{eqn:constrain td problem-governing eqn}\\
    &[f_0,f_2,\dots,f_{2K}]^T=P[f_0,f_2,\dots,f_{2K}]^T\label{eqn:contraint td problem-constrain}
  \end{empheq}
\end{subequations}
where (\ref{eqn:constrain td problem-governing eqn}) is rewritten from (\ref{eqn:coefficient-funciton-itp}). The projection step (\ref{eqn:contraint td problem-constrain}) is to guarantee that (\ref{eqn:stationary-wigner-eqn-hermite-1D}) holds. In next subsection, we will propose a reconstruction method to deal with the underdetermined system (\ref{eqn:constrain td problem-governing eqn}) and utilize the projection (\ref{eqn:contraint td problem-constrain}).

\subsection{Reconstruction}\label{subsec:reconstruction}
Notice that (\ref{eqn:constrain td problem-governing eqn}) depicts a underdetermined system due to the appearance of $f_{2K+2}$, to close the system we consider the reconstruction of $f_{2K+2}$. As discussed at the beginning of this section, $f_{2K+2}$ is determined by $\{f_{2k}\}_{k=0}^K$ from stationary Wigner equation (\ref{eqn:stationary-wigner-eqn-hermite-1D}). On the other hand, stationary Wigner equation is a constraint to be met for ground state calculation. With appropriate boundary condition, we consider the reconstructing $f_{2K+2}$ based on (\ref{eqn:stationary-wigner-eqn-hermite-1D}). To achieve a system with finite unknowns, we adopt truncation on our domain, i.e., use $[-a,a]$ instead of $\mathbb{R}$, where $a$ is a large positive number. Since $\lim_{x\rightarrow\pm\infty}f_{2k}=0$ as previous discussion, it is reasonable to use the approximation $f_{2K+2}(-a)=0$ as boundary condition for reconstruction.

It is deserved to mention that such reconstruction is not confined to $f_{2K+2}$. Since the approximation $f_{2k+2}(-a)=0$ also works for $0\leqslant k<K$, start from $f_0$, we can construct $f_{2k}$ for $k=1,2,\dots,K$ in turn, and finally obtain coefficient functions $\{f_{2k}\}_{k=0}^K$ satisfying (\ref{eqn:contraint td problem-constrain}). Therefore, such reconstruction is also adopted for the utilization of (\ref{eqn:contraint td problem-constrain}). 

\subsection{Numerical discretization}
Since the Crank-Nicolson scheme is unconditionally stable and second-order in time, it is adopted for temporal discretization as following approximation
\begin{equation}\label{eqn:eHcf-approximation}
  \text{e}^{-\Delta\tau H_{\rm cf}}
  \approx\frac{1+\Delta\tau H_{\rm cf}/2}{1-\Delta\tau H_{\rm cf}/2}.
\end{equation}
With regard to spatial discretization, we consider finite element method (FEM). The coefficient functions $\{f_{2k}\}_{k=0}^K$ are approximated by the linear combination of piecewise-polynomial basis functions $\{\phi_i\}$, where $\{\phi_i\}$ is defined on a set of real space interpolation nodes $\{p_i\}$. Denote coefficients by $\{\varphi_{k,i}\}$, then we have:
\begin{equation}
  f_{2k}\approx\sum\limits_{i=1}^{N_{\rm basis}}\varphi_{k,i}\phi_i,~~~k=0,1,\dots,K,
\end{equation}
where $N_{\rm basis}$ stands for the dimension of space $V_h=\spn\{\phi_i,1\leqslant i\leqslant N_{\rm basis}\}$. And $\phi_i$ is the $i$-th basis function that is typically chosen such that $\phi_i(p_j)=\delta_{i,j}$, $\delta_{i,j}$ is Kronecker delta function. As a result, we have $\varphi_{k,i}=f_{2k}(p_i)$, for $1\leqslant i\leqslant N_{\rm basis}$ and $0\leqslant k\leqslant K$. Let the superscript $(n)$ denote the corresponded term at time $\tau=n\Delta\tau$, where $\Delta\tau$ is the time step, and denote $\varphi_k=[\varphi_{k,1},\varphi_{k,2},\dots,\varphi_{k,N_{\rm basis}}]^T$. Introducing the approximation (\ref{eqn:eHcf-approximation}) to finite element discretization of the equation (\ref{eqn:coefficient-funciton-itp}) within the subspace $V_h$, we obtain
\begin{equation}\label{eqn:coefficient-function-itp-discritized}
  \begin{aligned}
    M\varphi_k^{(n+1)}-\frac{1}{2}\Delta\tau H_{\rm cf}^{\rm FEM}(\varphi_0^{(n+1)},\varphi_1^{(n+1)},\dots,\varphi_{k}^{(n+1)},\varphi_{k+1}^{(n+1)})
    =M\varphi_{k}^{(n)}+\frac{1}{2}\Delta\tau H_{\rm cf}^{\rm FEM}(\varphi_0^{(n)},\varphi_1^{(n)},\dots,\varphi_{k}^{(n)},\varphi_{k+1}^{(n)}),
  \end{aligned}
\end{equation}
where
\begin{equation}\label{eqn:Hcf-coefficient-function-itp-discritized}
  \begin{aligned}H_{\rm cf}^{\rm FEM}(\varphi_0^{(n)},\varphi_1^{(n)},\dots,\varphi_{k}^{(n)},\varphi_{k+1}^{(n)})
    =(2k+2)(2k+1)M\varphi_{k+1}^{(n)}+\left(\frac{1}{4}S+(4k+1)M\right)\varphi_{k}^{(n)}+M\varphi_{k-1}^{(n)}
    +\sum\limits_{l=0}^k\frac{2}{(2l )!(2i)^{2l}}M^{(2l )}\varphi_{k-l}^{(n)}.
  \end{aligned}
\end{equation}
In the equations above,
\begin{equation}\label{eqn:matrix-coefficient-function-itp-discritized}
  S=\left[\int_\Omega\nabla\phi_i\cdot\nabla\phi_jdx\right],
  ~~~M=\left[\int_\Omega\phi_i\phi_jdx\right],
  ~~~M^{(m )}=\left[\int_\Omega\frac{\partial^{m}V}{\partial x^{m}}\phi_i\phi_jdx\right],
\end{equation}
and $\Omega=[-a,a]$ is the truncated domain. The finite element discretization of (\ref{eqn:stationary-wigner-eqn-hermite-1D}) yields the equations for reconstruction
\begin{equation}\label{eqn:reconstruct}
  -(2k+2)A\varphi_{k+1}^{(n)}+A\varphi_{k}^{(n)}+\sum\limits_{l =0}^k\frac{1}{(2l+1)!(2i)^{2l}}M^{(2l+1)}\varphi_{k-l}^{(n)}=0,
  ~~~k=0,1,\dots,K,
\end{equation}
where
\begin{equation}
  A=\left[\int_\Omega\frac{\partial\phi_i}{\partial x}\phi_jdx\right]
\end{equation}
and $M^{(m)}$ is defined as above. While the normalization condition of $\{\varphi_k\}_{k=0}^K$
\begin{equation}
    \int f_0 dx=1
\end{equation}
becomes
\begin{equation}\label{eqn:vphi-normalization-condition}
  \sum\limits_{i=0}^{N_{\rm basis}}\varphi_{0,i}\int_\Omega\phi_idx=1.
\end{equation}
The main components of the flow chat are introduced as follows

\begin{algorithm}[H]\label{alg:1}
	\caption{One-dimensional ground state calculation of Wigner function based on ITP}
	\KwData{Truncation order $K$, time $T$, time step $dt $, tolerance $tol $, initial value $\{\varphi_{k}^{\rm ini}\}_{k=0}^K$.}
	\KwResult{Coefficient functions of Wigner function of ground state, represented by $\{\varphi_k^{\rm now}\}_{k=0}^K$.}
	$\varphi_k^{\rm now}\leftarrow\varphi_k^{\rm ini}$,	$\varphi_k^{\rm last}\leftarrow\varphi_k^{\rm ini}$ for $0\leqslant k\leqslant K$\;
	\While{$t<T$ and $err>tol $}{
		Reconstruct $\varphi_{K+1}^{\rm now}$ by the approximation $f_{2K+2}(-a)=0$ and (\ref{eqn:stationary-wigner-eqn-hermite-1D})\;
		Time propagation based on ITP method (\ref{eqn:coefficient-function-itp-discritized}), (\ref{eqn:Hcf-coefficient-function-itp-discritized}) and (\ref{eqn:matrix-coefficient-function-itp-discritized})\;
		Apply strategy to restrict solution space, i.e., set $f_0$ to be even (optional)\;
		Orthogonalization process by (\ref{eqn:orthogonality}) (optional for calculation of excited state)\;
		Normalize $\{\varphi_{k}^{\rm now}\}_{k=0}^K$ based on (\ref{eqn:vphi-normalization-condition})\;
		\For{$k=1:K$}{
			Construct $\varphi_{k}^{\rm now}$ by $\{\varphi_{l}^{\rm now}\}_{l =0}^{k-1}$ using the same method to construct $\varphi_{K+1}^{\rm now}$\;
		}
		Calculate $err=\max_{0\leqslant k\leqslant K}\|\varphi_{k}^{\rm now}-\varphi_{k}^{\rm last}\|_\infty$, $t=t+dt $, $\varphi_{k}^{\rm last}=\varphi_{k}^{\rm now}$ for $0\leqslant k\leqslant K$\;
	}
\end{algorithm}
\begin{remark}\rm
  As indicated in (\ref{eqn:wigner-definition}) that Wigner function depends on only one variable, thus an arbitrary two-variable function cannot be expressed by linear combination of the Wigner functions of pure states. A straightforward idea is to consider the strategy to restrict solution space. Based on our numerical experiment, although the convergence to ground state can be observed without such strategy, it can accelerate the convergence process. Furthermore, this strategy becomes necessary if the excited state is taken into consideration.
  
  Particularly, when we have a symmetric potential, 
  it can be easily proved that the Schr\"odinger eigenfunction is either odd or even, which always leads to an even density. Therefore artificially taking $f_0$ to be even is a reasonable choice to restrict the solution space in such cases.
\end{remark}

\section{Numerical experiments}\label{sec:numerical result}

In this section, the numerical effectiveness of our method is verified
by following four examples. The first example is a quantum harmonic
oscillator, which is a fundamental example employed as a sanity check. Meanwhile, the influence of size of domain for the
simulation is studied in this example, and the ability of calculating the
excited states using our approach is also demonstrated by supplementing Algorithm \ref{alg:1} with additional orthogonalization
process. Second, we consider a one-dimensional
hydrogen system, where the solution contains singularity, and thus we need at least
quadratic finite elements to produce a reliable ground
state solution. The third and fourth examples are, respectively, a
two-particle Hooke's atom system and its
variation, where both systems are represented
in the context of the density functional theory. Numerical results from
these two examples show the potential of our method for simulating
large scale systems.

In all examples, the computational domain is chosen as 
$[-10,10]$ with a uniform mesh unless otherwise specified. In each example, different mesh sizes
  $h$ and truncation orders $K$ are tested to show the convergence rate
  and the influence of truncation order. Error estimation is given by
  the infinity norm.

\subsection{Harmonic oscillator}\label{eg:1}

We consider the potential of harmonic oscillator
\begin{equation}
  V=\frac{1}{2}\omega^2x^2.
\end{equation}
For the sake of simplicity, we take $\omega=1$. Then the Wigner function of the $n$-th eigenstate is given by \cite{groenewold1946}
\begin{equation}
  f^{(n)}(x,p)
  =\frac{(-1)^n}{\pi}\exp\left[-2\left(\frac{p^2}{2}+\frac{x^2}{2}\right)\right]L_n\left(4\left(\frac{p^2}{2}+\frac{x^2}{2}\right)\right),
  ~~~\text{with energy}~~~E_n=\frac{2n+1}{2},
\end{equation}
where $L_n(x)$ is the $n$-th Laguerre polynomial.

We first discuss the choices of truncation on domain. The numerical
results obtained from different sizes of domain are demonstrated in
Figure \ref{fig:harmonic-oscillator-truncation}.
\begin{figure}[H]
  \centering
  \includegraphics[width=.45\linewidth]{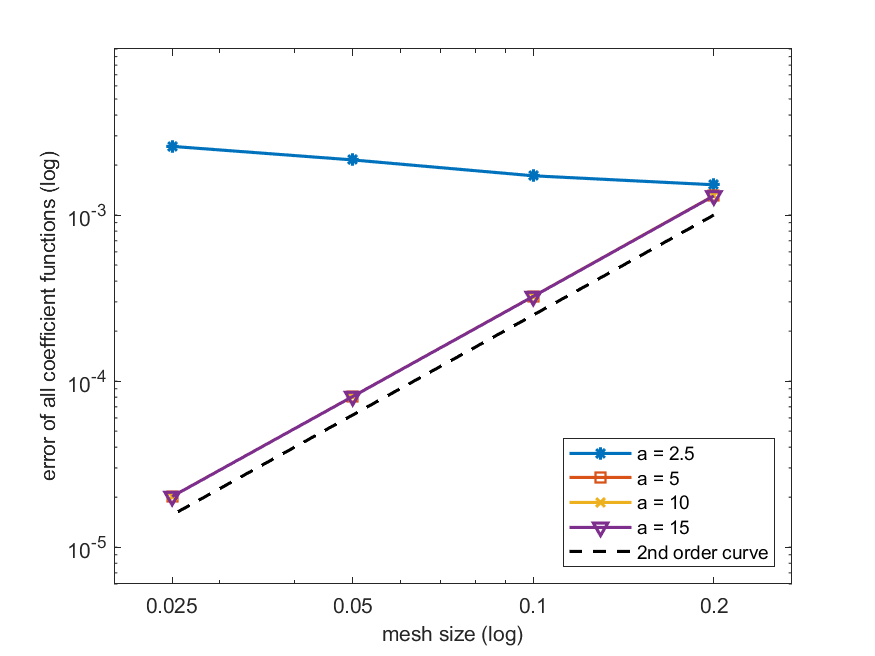}
  \includegraphics[width=.45\linewidth]{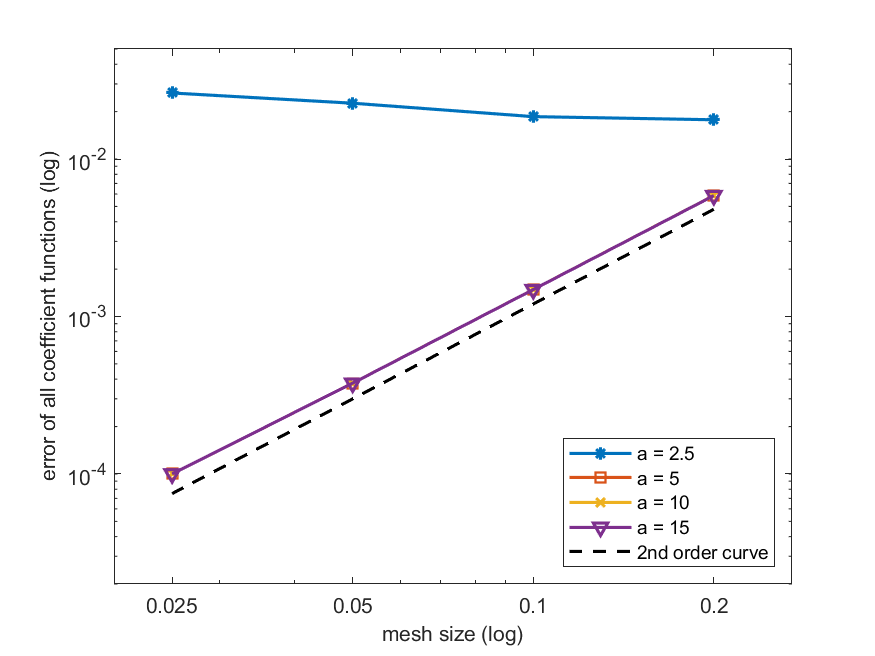}
  \caption{Errors $\max_{0\leqslant k\leqslant K}\|f_{2k}-f_{2k}^{\rm exc}\|_\infty$ for different truncated domain with $K=10$ for ground state (left) and 1st excited state (right).}
  \label{fig:harmonic-oscillator-truncation}
\end{figure}

It can be observed from Figure
\ref{fig:harmonic-oscillator-truncation} that with a small domain $a = 2.5$,
the desired numerical convergence cannot be observed. However, with
the increment of the domain size, theoretical convergence rate of
linear finite element method can be successfully obtained with
$a=5,10,15$. To balance numerical accuracy and efficiency, in our
following simulations, $a = 10$ is always used.

Next, we would like to show the difference between the algorithms with/without the projection step (\ref{eqn:contraint td problem-constrain}).
It follows the discussion in the previous sections that introducing the constraint of stationary Wigner equation helps accelerate the convergence, as is also confirmed by our numerical experiments.
It is illustrated in the Figure \ref{fig:ho-compare-construction} that the error of the first excited state fails to decrease to $1.0\times10^{-10}$ even with sufficiently long time. Based on our numerical experiment, it is worth mentioning that the constraint of the stationary Wigner equation also contributes to suppressing some possible numerical or modelling errors accumulated during the simulation. In fact, divergent results have been observed in our experiments without such a constraint.

\begin{figure}[H]
    \centering
    \includegraphics[width=.45\linewidth]{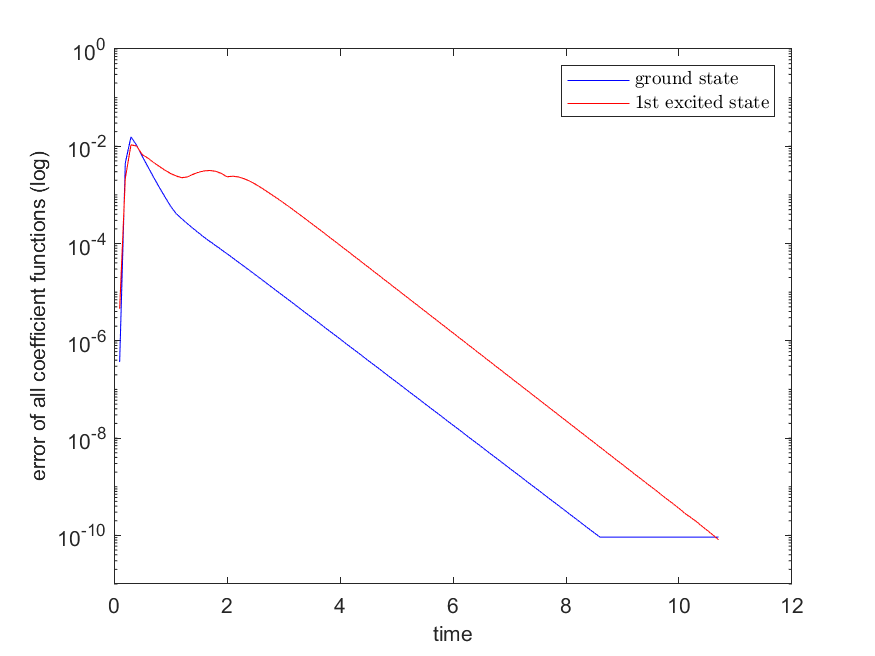}
    \includegraphics[width=.45\linewidth]{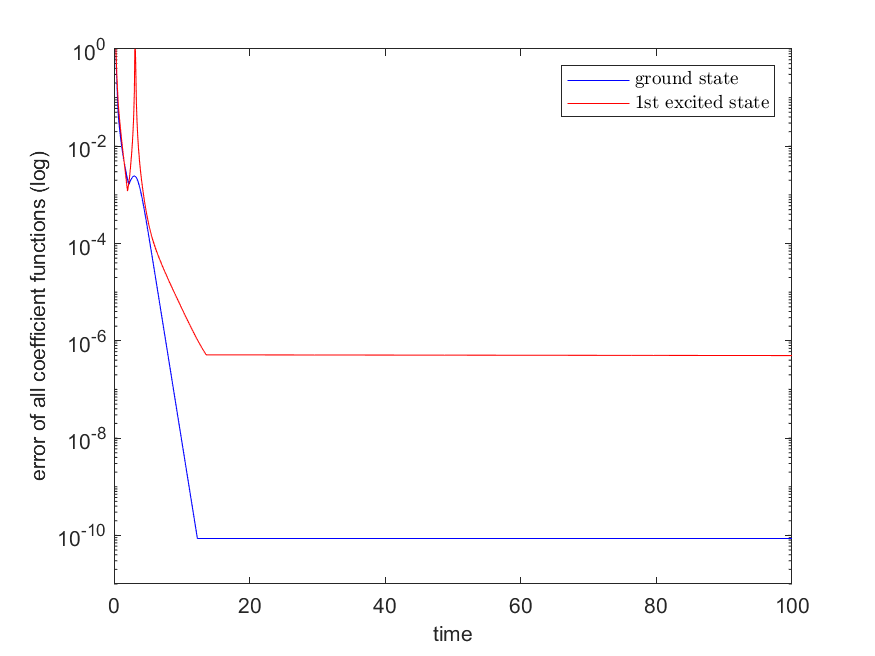}
    \caption{Decay of errors $\max_{0\leqslant k\leqslant K}\|f_{2k}^\textrm{now}-f_{2k}^\textrm{last}\|_\infty$ for numerical simulation with(left)/without(right) constraint of stationary Wigner equation.}
    \label{fig:ho-compare-construction}
\end{figure}

Besides the ground state, it is noted that by adding an
orthogonalization process in Algorithm \ref{alg:1}, excited states of
the system can be calculated simultaneously. 
Furthermore, the
effect of truncation order $K$ for the numerical solution is also
studied in our numerical experiments. The numerical results are
provided in Figure \ref{fig:harmonic-oscillator-result}.

\begin{figure}[H]
  \centering
  \includegraphics[width=.32\linewidth]{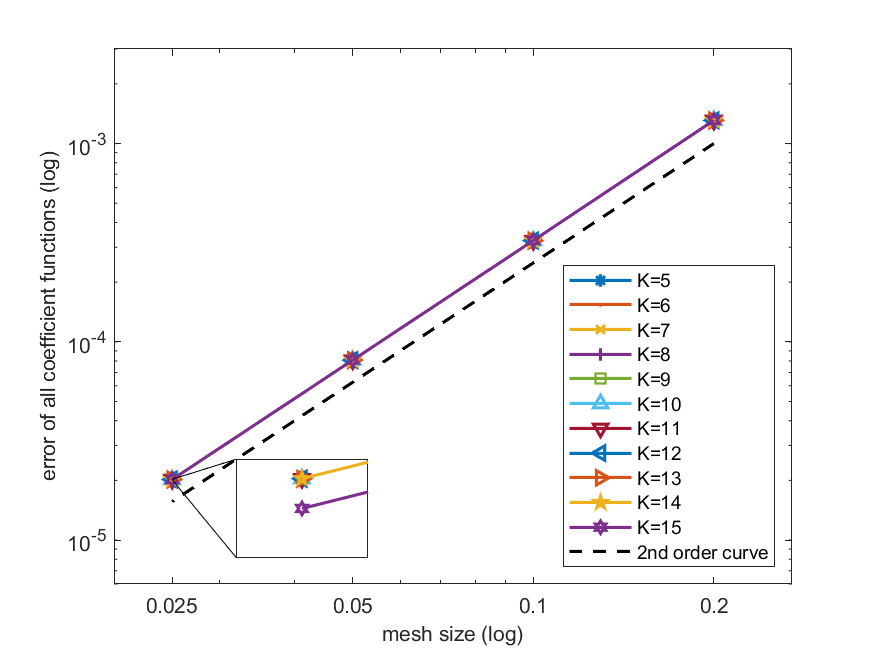}
  \includegraphics[width=.32\linewidth]{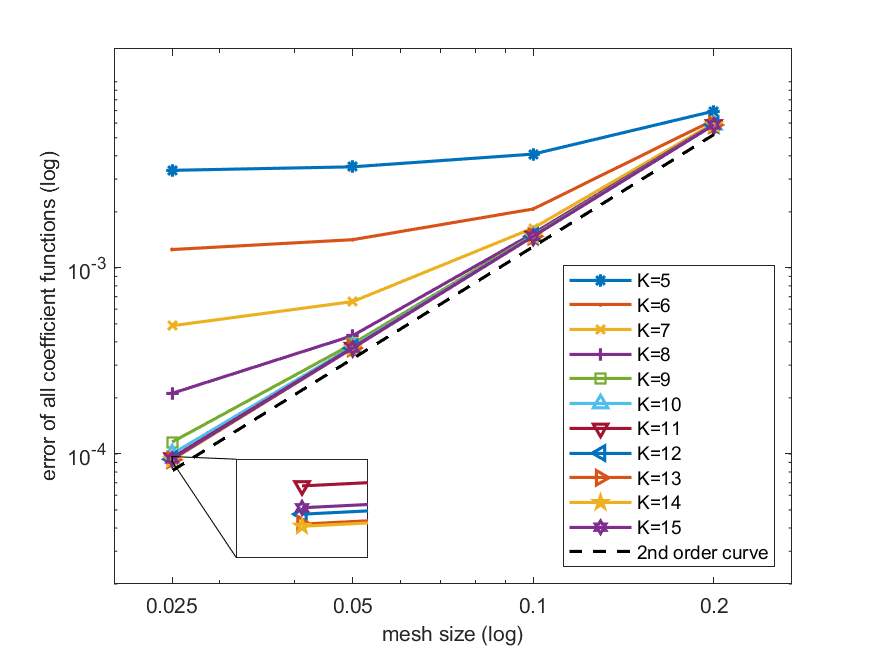}
  \includegraphics[width=.32\linewidth]{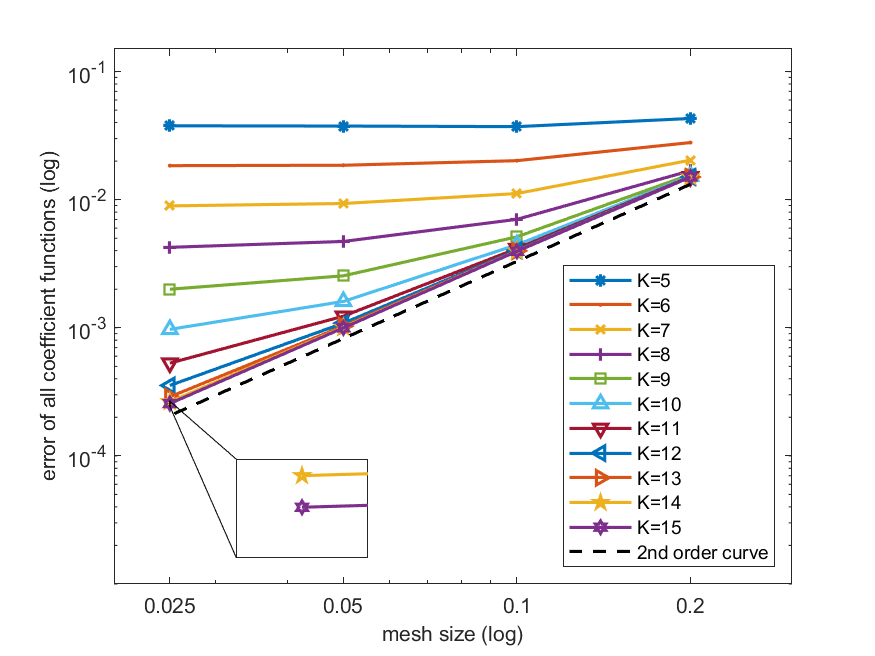}
  \caption{Errors $\max_{0\leqslant k\leqslant K}\|f_{2k}-f_{2k}^{\rm
      exc}\|_\infty$ for the ground state, 1st excited state and 2nd
    excited state, respectively (from left to right).}
  \label{fig:harmonic-oscillator-result}
\end{figure}
The numerical observations of Figure
\ref{fig:harmonic-oscillator-result} are summarized as follows: i) In the computation of all the states, second-order convergence with respect to the mesh size is observed. ii) The calculation of ground state converges fast with respect to $K$;
    our numerical result for $K = 5$ almost
    coincides with that for $K=15$. The convergence is slower for
    the excited state calculation, but the behavior still looks like the spectral convergence due to the smoothness of the solution.
    iii) Larger truncation order $K$ is needed when higher excited
    state is calculated. This might be caused by the accumulation of
    numerical error from the orthogonalization process. In addition,
    since the 2nd excited state has higher energy, more terms in
    Hermite expansion are needed to depict its more oscillatory
    behavior.

\subsection{One-dimensional hydrogen}\label{eg:2}

The one-dimensional hydrogen system consists of an electron moving in
the one-dimensional potential
\begin{equation}
  V(x)=-\frac{1}{|x|},
\end{equation}
with boundary condition $\rho(x)=0$. It is shown in \cite{loudon1959}
that the ground state density is
\begin{equation}
  \rho_0(x)=2x^2e^{-2|x|}~~~\text{with energy}~~~E_0=-\frac{1}{2}.
\end{equation}
It follows (\ref{eqn:wigner-definition}) that its Wigner function reads
\begin{equation}
  \begin{aligned}
    f(x,p)
    =2e^{-2|x|}\left(x^2\delta(p)+\frac{1}{4}\delta''(p)\right).
  \end{aligned}
\end{equation}
In our numerical experiment, it is found that linear finite element
method fails to provide a convergent result, which can be explained as
follows. 

Substituting $f_0=2x^2\textrm{e}^{-2|x|}+\varepsilon_0$ into the following numerical construction of $f_2$ we find
\begin{equation}\label{eqn:f2-construction}
  \begin{aligned}
    2\frac{\partial f_2}{\partial x}+\frac{\partial f_0}{\partial x}
    &=-\frac{\partial V}{\partial x}f_0
    =\frac{\sgn(x)}{x^2}\cdot\bigg{(}2x^2e^{-2|x|}+\varepsilon_0\bigg{)}
    =2\sgn(x)e^{-2|x|}+\frac{\sgn(x)}{x^2}\cdot\varepsilon_0,
  \end{aligned}
\end{equation}
Therefore
in the element containing original point, we have the numerical error
$\varepsilon_0/x^2=O(1)$ which prevents the numerical convergence to the exact solution. To tackle this problem, we consider the quadratic
finite element basis functions in this example, which yields
theoretical error $O(h )$ by (\ref{eqn:f2-construction}).
\begin{figure}[H]\label{fig:hydrogen}
  \centering
  \includegraphics[width=.5\linewidth]{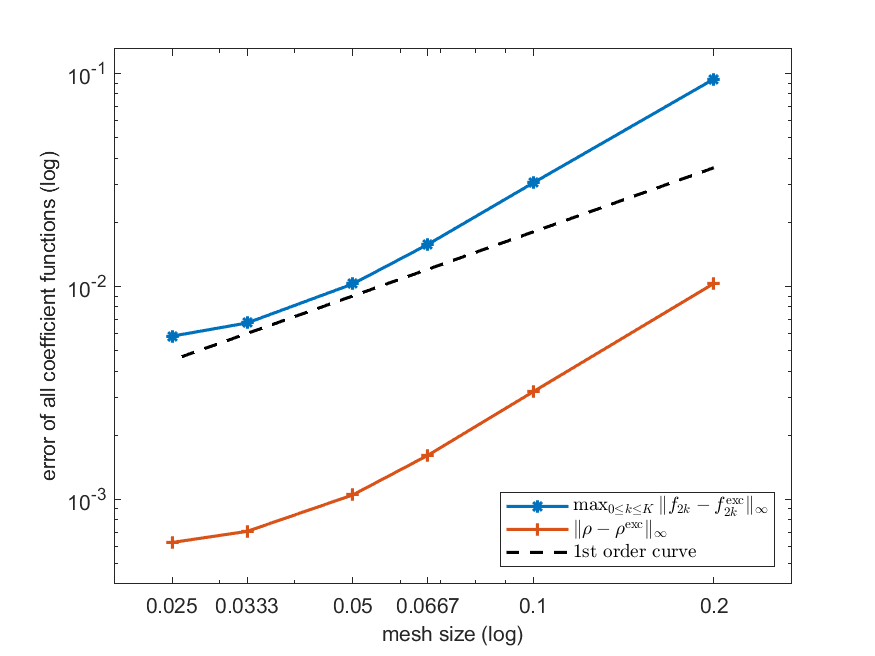}
  \caption{Numerical errors of one-dimensional hydrogen when $K=1$.}
\end{figure}

Numerical results are shown in Figure \ref{fig:hydrogen}, where we
present both the error of all coefficient functions and the error of
density. The following observations can be made from the result: i) The error of density (orange curve) is better than that of all
    coefficient functions (blue curve), as the normalization condition
    acts on density.
    ii) It is observed that when the mesh size is large, the convergence rate is
    higher than the theoretical one. As the mesh is refined, the error gradually saturates due to the truncation of the infinite system, and eventually the overall error is dominated by the reconstruction of $f_4$ when the mesh size is sufficiently small.
    Note that for larger $K$, the higher-order derivatives of $V$ turn out to be more singular, leading to the requirement of higher-order numerical methods to guarantee the convergence. 
    

\subsection{Hooke's atom}\label{eg:3}

In this example we consider the Hooke's atom consisting of two
electrons oscillating in the parabolic well. The two-particle
Schr\"odinger equation is
\begin{equation}\label{eqn:Psi-two-electron-parabolic-well}
  \left(-\frac{1}{2}\frac{\partial^2}{\partial x_1^2}-\frac{1}{2}\frac{\partial^2}{\partial x_2^2}+V_{\rm ext}(x_1)+V_{\rm ext}(x_2)+\frac{1}{|x_1-x_2|}\right)\Psi(x_1,x_2)=E\Psi(x_1,x_2),
\end{equation}
where the external potential of Hooke's atom is
\begin{equation}\label{eqn:Vext-two-electron-parabolic-well}
  V_{\rm ext}(x)=\frac{1}{2}k_{\rm Hooke}x^2,~~~k_{\rm Hooke}=\frac{1}{4}.
\end{equation}
It follows the derivation in Appendix B that the ground state density is
\begin{equation}
  \begin{aligned}
    \rho(x)
    =2\int|\Psi(x,x_2)|^2dx_2
    =2C^2e^{-\frac{x^2}{2}}\bigg{(}
    (4+2x^2)e^{-\frac{x^2}{2}}
    +\sqrt{2\pi}\left(\frac{7}{4}+\frac{5}{2}x^2+\frac{1}{4}x^4\right)
    +\sqrt{2\pi}\erf\left(\frac{x}{\sqrt{2}}\right)(3x+x^3)
    \bigg{)},
  \end{aligned}
\end{equation}
where $C=(16\sqrt{\pi}+10\pi)^{-1/2}$, and $\erf(x)$ is the error function defined as
\begin{equation}
  \erf(x)=\frac{2}{\sqrt{\pi}}\int_0^xe^{-y^2}dy.
\end{equation}
Since $\erf(x)$ is odd by its definition, $\rho(x)$ is an even
function.

This two-particle system can be transformed into the Kohn-Sham
equation
\begin{equation}\label{eqn:psi-Kohn-Sham}
  \left(-\frac{1}{2}\frac{d^2}{dx^2}+V_{\rm KS}\right)\psi_i=\varepsilon_i\psi_i(x),~~~i=1,2,
\end{equation}
where
\begin{equation}\label{eqn:VKS}
  V_{\rm KS}=V_{\rm ext}+V_{\rm H}+V_{\rm XC}.
\end{equation}
The Kohn-Sham orbitals of (\ref{eqn:psi-Kohn-Sham}) are given by
\begin{equation}
  \psi_i(x)=\sqrt{\rho(x)/2},~~~i=1,2.
\end{equation}
With $\psi=\psi_1=\psi_2$ and
$\varepsilon=\varepsilon_1=\varepsilon_2$, and with the assumption
that the functional derivative of $V_{\rm H}+V_{\rm XC}$ vanishes at
infinity, the Kohn-Sham potential can be exactly expressed as
\cite{pinchus1986}
\begin{equation}
  \begin{aligned}
    V_{\rm KS}(x)
    &=\varepsilon+\frac{1}{2}\frac{d^2\psi(x)/dx^2}{\psi(x)}
    =\varepsilon\texttt{}+\frac{\rho''(x)\rho(x)-\left(\rho'(x)\right)^2/2}{4\rho(x)^2}.
  \end{aligned}
\end{equation}
Since $\rho(x)$ is even, $V_{\rm KS}(x)$ is also an even function.

The numerical results for $K=0,1,2$ with our method
are listed in following table.
\begin{table}[H]
	\centering
	\caption{Numerical errors and convergence rate of Hooke's atom.}
	\label{tab:two-particle}
	\begin{tabular}{c|cc|cc|cc}\hline
		$K$	&\multicolumn{2}{|c}{0}	&\multicolumn{2}{|c}{1}	&\multicolumn{2}{|c}{2}\\\hline
		$h$	&error	&order	&error	&order	&error	&order\\\hline
		~	&\multicolumn{6}{|c}{$\max_{0\leqslant k\leqslant K}\|f_{2k}-f_{2k}^{\rm exc}\|_\infty$}\\\hline
		2.0e-01	&1.4667e-04 &-
		&1.1195e-03 &-	
		&1.1194e-03 &-\\
		1.0e-01	&3.4544e-05 &2.0861
		&2.5899e-04 &2.1119
		&2.5898e-04 &2.1118\\
		5.0e-02	&8.9058e-06 &1.9556
		&6.3706e-05 &2.0234
		&6.3705e-05 &2.0234\\
		2.5e-02	&2.5259e-06 &1.8179
		&1.5864e-05 &2.0057			
		&1.5864e-05 &2.0057\\\hline
		~	&\multicolumn{6}{|c}{$\|f_0-\rho^{\rm exc}\|_\infty$}\\\hline
		2.0e-01	&- &-
		&1.5784e-04 &-
		&1.5785e-04 &-\\
		1.0e-01	&- &- 
		&3.8595e-05 &2.0320
		&3.8597e-05 &2.0320\\
		5.0e-02	&- &- 
		&9.5975e-06 &2.0077
		&9.5980e-06 &2.0077\\
		2.5e-02	&- &- 
		&2.3961e-06 &2.0020
		&2.3963e-06 &2.0019\\\hline
	\end{tabular}
\end{table}

It can be observed in Table \ref{tab:two-particle} that i) the numerical accuracy of the density is better than
    the high-order coefficients which is similar to the previous example.
  ii) For $K=0$, the numerical convergence towards the theoretical
    result can already be obtained. However, with the refinement of mesh grids,
    the convergence order starts to decrease, since for small $K$, the reconstruction error dominates when
    mesh size is sufficiently small.
    iii) The above issue is remedied when $K$ becomes larger. For
    $K=1,2$, the theoretical convergence order with respect to the grid size is obtained 
    in all the simulations. Furthermore,
    comparable numerical accuracy can be observed from results with
    $K=1$ and $K=2$, which means that in this example, a small $K$ can already provide
    sufficient numerical accuracy.

\subsection{Contact-interacting Hooke's atom}\label{eg:4}

Now we consider a more practical example by replacing the interacting
function in (\ref{eqn:Psi-two-electron-parabolic-well}) with delta
function
\begin{equation}\label{ci Hooke}
  \left(-\frac{1}{2}\frac{\partial^2}{\partial x_1^2}-\frac{1}{2}\frac{\partial^2}{\partial x_2^2}+V_{\rm ext}(x_1)+V_{\rm ext}(x_2)+\delta(x_1-x_2)\right)\Psi(x_1,x_2)=E\Psi(x_1,x_2),
\end{equation}
where $V_{\rm ext}(x)$ is defined as
(\ref{eqn:Vext-two-electron-parabolic-well}). This two-particle system
can be also transformed into the Kohn-Sham equation as
(\ref{eqn:psi-Kohn-Sham}) with (\ref{eqn:VKS}). The energy of
system is given by \cite{FNM03}
\begin{equation}
  E=\sum\limits_{i=1}^2\varepsilon_i-U_{\rm H}[\rho]-\int V_{\rm XC}([\rho];x)dx+E_{\rm XC}[\rho],
\end{equation}
where
\begin{equation}
  U_{\rm H}[\rho]=\frac{1}{2}\int\rho(x)^2dx,\mbox{ and } E_{\rm
    X}[\rho]=-\frac{1}{4}\int\rho(x)^2dx.
\end{equation}
The local-density correlation energy functional is \cite{magyar2004}
\begin{equation}
  E_{\rm C}^{\rm LDA}[\rho]
  =\int\left(\frac{a\rho(x)^3+b\rho(x)^2}{\rho(x)^2+d\rho(x)+e}\right)dx,
\end{equation}
where $a=-1/24$, $b=-0.00436143$, $d=0.252758$ and $e=0.0174457$.
By these definitions, we have
\begin{equation}
  V_{\rm H}([\rho];x)=\frac{\delta U_{\rm H}[\rho]}{\delta\rho(x)}
  ~~~\text{and}~~~
  V_{\rm XC}([\rho];x)=\frac{\delta (E_{\rm X}[\rho]+E_{\rm C}^{\rm LDA}[\rho])}{\delta\rho(x)}.
\end{equation}

To solve (\ref{ci Hooke}), a self-consistent iteration is employed due
to the nonlinearity\cite{book-A_Primer}. To validate our method, the numerical result of
the Schr\"odinger wave function is calculated as a reference. Both the
numerical results of our method with $K=0,1,2$ and the ones of
Schr\"odinger wave function are listed in Table
\ref{tab:interacting-hooke-comparison}. Here both the Wigner function and the wave function are solved on different grids ranging from $h = 0.2$ to $h = 6.25 \times 10^{-3}$, and the ``energy'' columns list the ground state energy for all our simulations. For the ``error'' columns, we list the difference of the particle densities $\rho_w$ and $\rho_s^{\mathrm{finest}}$, where $\rho_w$ stands for the Wigner function computed using different parameters $h$ and $K$, and $\rho_s^{\mathrm{finest}}$ is obtained from the wave function computed on the finest grid.

\begin{table}[H]
  \centering
  \caption{Numerical errors of density $\|\rho_w-\rho_s^{\rm finest}\|_\infty$ and energy of contact-interacting Hooke's atom.}
  \label{tab:interacting-hooke-comparison}
  \begin{tabular}{c|cc|cc|cc|c}\hline
    $K$	&\multicolumn{2}{|c}{0}	&\multicolumn{2}{|c}{1}	&\multicolumn{2}{|c|}{2} &Schr\"odinger\\\hline
    $h$	&error	&energy	&error	&energy	&error	&energy &energy\\\hline
    2.00e-01	&2.1068e-03 &1.31718
    &1.9710e-03 &1.31727
    &2.2879e-03 &1.31568
    &1.31476\\
    1.00e-01	&5.8217e-04 &1.31398
    &5.1479e-04 &1.31403
    &6.0338e-04 &1.31360
    &1.31347\\
    5.00e-02	&1.7063e-04 &1.31323
    &1.3680e-04 &1.31326
    &1.5979e-04 &1.31315
    &1.31315\\
    2.50e-02	&5.3808e-05 &1.31306
    &3.6827e-05 &1.31308
    &4.2656e-05 &1.31305
    &1.31307\\
    1.25e-02 &1.7723e-05 &1.31304
    &9.2112e-06 &1.31304
    &1.0677e-05 &1.31304
    &1.31305\\
    6.25e-03 &5.2688e-06 &1.31303
    &3.0021e-06 &1.31304
    &3.4460e-06 &1.31304
    &1.31304\\\hline
  \end{tabular}
\end{table}

It is shown in Table \ref{tab:interacting-hooke-comparison} that
i) for each fixed $K$, the convergence towards the Schr\"odinger results can be observed as the mesh is refined;
ii) results of energies for two cases
coincide with each other very well, especially when mesh size is
sufficiently small ($h\le 0.025$). These observations again validate both
the model and the numerical method proposed in this paper.

\begin{table}[H]
  \centering
  \caption{Relative errors $\max_{0\leqslant k\leqslant
      K}\|f_{2k}-f_{2k}^{\rm finest}\|_\infty$ and convergence order
    of contact-interacting Hooke's atom.}
  \label{tab:interacting-hooke-convergence-order}
  \begin{tabular}{c|cc|cc|cc}\hline
    $K$	&\multicolumn{2}{|c}{0}	&\multicolumn{2}{|c}{1}	&\multicolumn{2}{|c}{2}\\\hline
    $h$	&error	&order	&error	&order	&error	&order\\\hline
    2.00e-01	&2.1015e-03 &-
    &1.9700e-03 &-	
    &2.2866e-03 &-\\
    1.00e-01	&5.7690e-04 &1.8650
    &5.1378e-04 &1.9390
    &6.0201e-04 &1.9253\\
    5.00e-02	&1.6536e-04 &1.8027
    &1.3580e-04 &1.9197
    &1.5841e-04 &1.9261\\
    2.50e-02	&4.8540e-05 &1.7684
    &3.5820e-05 &1.9226			
    &4.1282e-05 &1.9401\\
    6.25e-03 &1.2454e-05 &1.9625
    &8.2042e-06 &2.1263
    &9.3032e-06 &2.1497 \\\hline
  \end{tabular}
\end{table}
In Table \ref{tab:interacting-hooke-convergence-order}, the relative
errors of numerical solution are demonstrated with a result obtained
on the finest mesh. 
It can be found that the convergence rate is
around the theoretical one of linear finite element method.

\section{Conclusions}\label{sec:conclusions}

In this paper, a model of Wigner function of eigenstate is proposed,
providing a theoretical foundation for ground state calculation in
Wigner formalism. With a simplified model of \cite{cai2013}, ITP
method is adopted for the realization of one-dimensional ground state
calculation. For validation purposes, we applied our method to the
simulation of several benchmark systems. The aim of first two
experiments is to show the capability of our method to handle excited
state and singularity of potential. While in the last two examples,
with the assistance of DFT, our approach successfully converges to the
ground state. In particular, the consistency of results from our
method with Schr\"odinger solution is observed in the last example,
indicating the feasibility of our method in DFT regime. 

Our ongoing work is to generalize the proposed method in this paper to
three-dimensional case, in which the reconstruction of coefficient
functions would be a nontrivial issue.

\section*{Acknowledgments}
The first author would like to thank the support from Macao PhD Scholarship (MPDS) from University of Macau. The second author was partially  supported by the Academic Research Fund of the Ministry of Education of Singapore under grant Nos. R-146-000-305-114 and R-146-000-291-114. The third author was partially supported by National Natural Science Foundation of China (Grant Nos. 11922120 and 11871489), MYRG of University of Macau (MYRG2019- 00154-FST) and Guangdong-Hong Kong-Macao Joint Laboratory for Data-Driven Fluid Mechanics and Engineering Applications (2020B1212030001).

\appendix
\section{Calculation of the coefficient in the inner product of two Wigner functions}
\label{app:coefficient}
First we consider the calculation of the integral
\begin{equation}
  I_n=\int x^n\exp(-x^2)dx.
\end{equation}
It is evident that $I_n=0$ if $n$ is odd. For nonzero even number $n$, using integration by part we find
\begin{equation}
  I_n=-\frac{1}{2}\int x^{n-1}d\exp(-x^2)
  =-\frac{1}{2}x^{n-1}\exp(-x^2)\bigg{|}_{-\infty}^\infty+\frac{n-1}{2}\int x^{n-2}\exp(-x^2)dx
  =\frac{n-1}{2}I_{n-2}.
\end{equation}
It is noted that $I_0=\sqrt{\pi}$, therefore we have
\begin{equation}
  I_n=\left\{\begin{array}{ll}
    \sqrt{\pi}n!/2^n(n/2)!, &\text{if $n$ is even;}  \\
    0 &\text{if $n$ is odd.} 
  \end{array}\right.
\end{equation}
Substituting it into the calculation of (\ref{eqn:C_aph-int-H_aphH_bet}) we obtain
\begin{equation}
\begin{aligned}
  C_{\alpha}&=\prod\limits_{j=1}^D\int\he_{\alpha_j}(x)\exp(-x^2)dx
  =\prod\limits_{j=1}^D\alpha_j!\sum\limits_{k=0}^{\alpha_j/2}\frac{(-1)^k}{k!(\alpha_j-2k)!}\frac{I_{\alpha_j-2k}}{2^k}\\
  &=\prod\limits_{j=1}^D\alpha_j!\sum\limits_{k=0}^{\alpha_j/2}\frac{(-1)^k}{k!(\alpha_j-2k)!}\frac{1}{2^k}\frac{\sqrt{\pi}(\alpha_j-2k)!}{2^{\alpha_j-2k}(\alpha_j/2-k)!}
  =\prod\limits_{j=1}^D\frac{\sqrt{\pi}\alpha_j!}{2^{\alpha_j}}\sum\limits_{k=0}^{\alpha_j/2}\frac{(-1)^k}{k!}\frac{1}{2^{-k}(\alpha_j/2-k)!}\\
  &=\prod\limits_{j=1}^D\frac{\sqrt{\pi}\alpha_j!}{2^{\alpha_j}}\frac{1}{(\alpha_j/2)!}\left(1-2\right)^{\alpha_j/2}
  =\prod\limits_{j=1}^D\frac{\sqrt{\pi}\alpha_j!}{(\alpha_j/2)!}\left(-\frac{1}{4}\right)^{\alpha_j/2}
  =\left(-\frac{1}{4}\right)^{|\alpha|/2}\frac{\pi^{D/2}\alpha!}{(\alpha/2)!}
\end{aligned}
\end{equation}
if all components of $\alpha$ are even; and $C_\alpha=0$ if one of the component of $\alpha$ is odd.
\section{Conversion between wave function and Wigner coefficient function}
\label{app:conversion}
Given the Wigner function of pure state
\begin{equation}\label{eqn:wigner-definition-pure-state}
  f(\mathbf{x},\mathbf{p})
  =\frac{1}{(2\pi)^D}\int\psi^*\left(\mathbf{x}-\frac{\mathbf{y}}{2}\right)\psi\left(\mathbf{x}+\frac{\mathbf{y}}{2}\right)\exp(-i\mathbf{p}\cdot\mathbf{y})d\mathbf{y}.
\end{equation}
We have expanded it by series involving Hermite polynomials in Section \ref{sec:hermite expansion}. It follows (\ref{eqn:hermite-polynomial-definition}) and (\ref{eqn:coefficient-function-calculation}) that we could define
\begin{equation}\label{eqn:coefficient-function-h-by-f}
  \begin{aligned}
    h_\mathbf{\alpha}(\mathbf{x})
    =\frac{1}{\mathbf{\alpha}!}\int \mathbf{p}^\mathbf{\alpha}f(\mathbf{x},\mathbf{p})d\mathbf{p}
    =\sum\limits_{0\leqslant2\mathbf{\beta}\leqslant\mathbf{\alpha}}\frac{1}{2^{|\mathbf{\beta}|}\mathbf{\beta}!}f_{\mathbf{\alpha}-2\mathbf{\beta}}(\mathbf{x}).
  \end{aligned}
\end{equation}
Conversely we have
\begin{equation}\label{eqn:coefficient-function-f-by-h}
  \begin{aligned}
    f_\mathbf{\alpha}(\mathbf{x})
    =\sum\limits_{0\leqslant2\mathbf{\beta}\leqslant\mathbf{\alpha}}\frac{(-1)^{|\mathbf{\beta}|}}{2^{|\mathbf{\beta}|}\mathbf{\beta}!}h_{\mathbf{\alpha}-2\mathbf{\beta}}(\mathbf{x}).
  \end{aligned}
\end{equation}
Thus the coefficient functions $\{f_\mathbf{\alpha}\}$ is equivalent to $\{h_\mathbf{\alpha}\}$ in the sense that one can be expressed by the linear combination of another.
Using (\ref{eqn:wigner-definition-pure-state}),  $h_\mathbf{\alpha}$ could be calculated by the derivative of wave function:
\begin{equation}\label{eqn:coefficient-function-h}
  \begin{aligned}
    h_\mathbf{\alpha}
    =\frac{1}{\mathbf{\alpha}!(2\pi)^D}\iint\psi^*\left(\mathbf{x}-\frac{\mathbf{y}}{2}\right)\psi\left(\mathbf{x}+\frac{\mathbf{y}}{2}\right)
    \frac{1}{(-i)^{|\mathbf{\alpha}|}}\frac{\partial^\mathbf{\alpha}}{\partial\mathbf{y}^\mathbf{\alpha}}\text{e}^{-i\mathbf{p}\cdot\mathbf{y}}d\mathbf{y}d\mathbf{p}
    =\frac{1}{\mathbf{\alpha}!(2i)^{|\mathbf{\alpha}|}}\sum\limits_{\mathbf{\beta}\leqslant\mathbf{\alpha}}(-1)^\mathbf{|\beta|}\binom{\mathbf{\alpha}}{\mathbf{\beta}}\frac{\partial^\mathbf{\beta}\psi^*}{\partial \mathbf{x}^\mathbf{\beta}}\frac{\partial^{\mathbf{\alpha}-\mathbf{\beta}}\psi}{\partial \mathbf{x}^{\mathbf{\alpha}-\mathbf{\beta}}}.
  \end{aligned}
\end{equation}
On the other hand, it follows (\ref{eqn:C_aph-int-H_aphH_bet}) and (\ref{eqn:hermite-polynomial-definition}) that
\begin{equation}
  \begin{aligned}
    \int\mathscr{H}_\mathbf{\alpha}(\mathbf{p})\exp\left(-\frac{1}{2}|\mathbf{p}-i(\mathbf{x}-\mathbf{x}_0)|^2\right)d\mathbf{p}
    =(2\pi)^{D/2}(i(\mathbf{x}-\mathbf{x}_0))^\mathbf{\alpha}.
  \end{aligned}
\end{equation}
Therefore with aid of (\ref{eqn:psi-inverse-wigner-transform}), (\ref{eqn:wigner-hermite-expansion}) and (\ref{eqn:hermite-expansion-basis-function}), we could recover the wave function from coefficient functions
\begin{equation}
  \begin{aligned}
    \psi(\mathbf{x})
    &=\frac{\exp(-|\mathbf{x}-\mathbf{x}_0|^2/2)}{\psi^*(\mathbf{x}_0)}\sum\limits_\mathbf{\alpha}f_\mathbf{\alpha}\left(\frac{\mathbf{x}+\mathbf{x}_0}{2}\right)(i(\mathbf{x}-\mathbf{x}_0))^\mathbf{\alpha}.
  \end{aligned}
\end{equation}
Thus wave function could also be reconstructed from coefficient functions. In addition, (\ref{eqn:coefficient-function-h}) and (\ref{eqn:coefficient-function-f-by-h}) provide a method to construct coefficient functions by wave function both theoretically and numerically.
\section{Ground state density of Hooke's atom}
Introducing change of variables $X=(x_1+x_2)/2$, $x=x_1-x_2$ we obtain
\begin{equation}
  \left(-\frac{1}{4}\frac{\partial^2}{\partial X^2}-\frac{\partial^2}{\partial x^2}+\frac{1}{4}X^2+\frac{1}{16}x^2+\frac{1}{|x|}\right)\Psi(x_1,x_2)=E\Psi(x_1,x_2)
\end{equation}
Taking $\Psi(x_1,x_2)=\Psi_X(X)\Psi_x(x)$, separating the variables we get
\begin{equation}\label{eqn:Psi-Xx-separating-variables}
  \left(-\frac{1}{4}\frac{\partial^2}{\partial X^2}+\frac{1}{4}X^2\right)\Psi_X=E_X\Psi_X,
  ~~~\text{and}~~~
  \left(-\frac{\partial^2}{\partial x^2}+\frac{1}{16}x^2+\frac{1}{|x|}\right)\Psi_x=E_x\Psi_x.
\end{equation}
The first equation is
$$\left(-\frac{1}{2}\frac{\partial^2}{\partial X^2}+\frac{1}{2}X^2\right)\Psi_X=2E_X\Psi_X,$$
which is the one-dimensional Schr\"odinger equation with harmonic oscillator potential. So the ground state wave function is
\begin{equation}\label{eqn:Psi-X-ground-state}
  \Psi_X=\pi^{-1/4}e^{-X^2/2},~~~\text{with energy}~~~2E_X=\frac{1}{2},E_X=\frac{1}{4}.
\end{equation}
Now we solve the equation for $\Psi_x$. At large $x$, the term $x^2/16$ dominates, which implies the approximate solution
$$\Psi_x(x)\approx Ae^{-x^2/8}+Be^{x^2/8}.$$
To find a normalizable solution, we take $B=0$, this suggests that $\Psi_x(x)=h(x)e^{-x^2/8}$. Then the Schr\"odinger equation becomes
\begin{equation}\label{eqn:h-two-partical-example}
  -\frac{d^2 h}{dx^2}+\frac{x}{2}\frac{dh}{dx}+\left(\frac{1}{4}+\frac{1}{|x|}-E_x\right)h=0.
\end{equation}
We look for solutions to (\ref{eqn:h-two-partical-example}) in the form of power series in $x$: $h(x)=\sum_{j=0}^\infty a_jx^j$. Plugging it into (\ref{eqn:h-two-partical-example}) we find
\begin{equation}
  \begin{aligned}
    -\sum\limits_{j=0}^\infty(j+1)(j+2)a_{j+2}x^j+\sum\limits_{j=1}^\infty\frac{1}{2}ja_jx^j+\sum\limits_{j=0}^\infty\left(\frac{1}{4}-E_x\right)a_jx^j+\sum\limits_{j=0}^\infty\sgn(x)a_{j+1}x^j+\frac{a_0}{|x|}=0.
  \end{aligned}
\end{equation}
Thus $a_0=0$, matching the coefficients we get
\begin{equation}\label{eqn:aj-relation-j1-2}
  -2a_2+\sgn(x)a_1=0
\end{equation}
and
\begin{equation}\label{eqn:aj-recursion-formula}
  a_{j+2}=\frac{\sgn(x)a_{j+1}+\left(j/2+1/4-E_x\right)a_j}{(j+1)(j+2)}
  ~~~\text{for}~~~j\geqslant1.
\end{equation}
For large $j=2k$, the recursion formula approximately becomes
$$a_{j+2}\approx\frac{1}{2j}a_j,~~~\text{with the approximate solution}~~~a_j\approx\frac{C(1/4)^{j/2}}{(j/2)!}=\frac{C(1/4)^k}{k!}$$
for some constant $C$. Similarly we find the approximate solution for odd $j=2k+1$
$$a_{j+2}\approx\frac{1}{2(j-1)}a_j~~~\Rightarrow~~~a_j\approx\frac{D(1/4)^{(j-1)/2}}{((j-1)/2)!}=\frac{D(1/4)^k}{k!}$$
for some constant $D$. These results yield the asymptotic behavior at large $x$:
$$h(x)=\sum\limits_k(a_{2k}x^{2k}+a_{2k+1}x^{2k+1})
\approx\sum\limits_k(C+Dx)\frac{(x^2/4)^k}{k!}
\approx(C+Dx)e^{x^2/4}.$$
But such behavior leads to $\Psi_x\approx (C+Dx)e^{x^2/8}$ at large $x$, which is not normalizable. Therefore the power series must terminate for some $j$. 
Suppose that $\{a_n\}$ terminates at $n$, i.e., $a_{n+1}=a_{n+2}=0$. Taking $j=n$ in (\ref{eqn:aj-recursion-formula}) we find
\begin{equation}
	\left(\frac{1}{2}n+\frac{1}{4}-E_x\right)a_j=0~~~\Rightarrow~~~E_x=\frac{2n+1}{4}.
\end{equation}
Note that when $n=1$, (\ref{eqn:aj-relation-j1-2}) implies trivial solution. For ground state we have $n=2$, the corresponding wave function and energy are
\begin{equation}\label{eqn:Psi-x-ground-state}
  \Psi_x=a_1\left(\frac{1}{2}\sgn(x)x^2+x\right)e^{-x^2/8}
  =a_1\left(\frac{1}{2}|x|+1\right)xe^{-x^2/8},
  ~~~\text{with energy}~~~
  E_x=\frac{5}{4}.
\end{equation}
Combining (\ref{eqn:Psi-X-ground-state}) and (\ref{eqn:Psi-x-ground-state}) we obtain the ground state wave function of Hooke's atom
\begin{equation}
  \Psi(x_1,x_2)=C\left(\frac{1}{2}|x|+1\right)xe^{-X^2/2}e^{-x^2/8}
  =C\left(\frac{1}{2}|x_1-x_2|+1\right)(x_1-x_2)e^{-(x_1^2+x_2^2)/4}
\end{equation}
with energy $E=E_X+E_x=3/2$, where $C=(16\sqrt{\pi}+10\pi)^{-1/2}$ is the normalization constant and $x,X$ are defined as above. Therefore the ground state density is given by
\begin{equation}
  \begin{aligned}
    &\rho(x)
    =2\int|\Psi(x,x_2)|^2dx_2
    =2C^2e^{-\frac{x^2}{2}}\bigg{(}
    (4+2x^2)e^{-\frac{x^2}{2}}
    +\sqrt{2\pi}\left(\frac{7}{4}+\frac{5}{2}x^2+\frac{1}{4}x^4\right)
    +\sqrt{2\pi}\erf\left(\frac{x}{\sqrt{2}}\right)(3x+x^3)
    \bigg{)},
  \end{aligned}
\end{equation}
where $\erf(x)$ is the error function defined as
\begin{equation}
  \erf(x)=\frac{2}{\sqrt{\pi}}\int_0^xe^{-y^2}dy.
\end{equation}

\bibliographystyle{plain}
\bibliography{ref}

\end{document}